\documentclass[smallextended]{svjour2}                    
\usepackage{graphicx}
 \usepackage{mathptmx}      
\usepackage{latexsym}
\usepackage{color}
\usepackage[colorlinks=true, citecolor=blue, filecolor=black, linkcolor=black, urlcolor=black]{hyperref}
\usepackage{cite}
\usepackage[normalem]{ulem}
\usepackage{enumerate}
\usepackage{latexsym}
\usepackage{hyperref}
\usepackage{bbm}
\usepackage{graphicx}
\usepackage{subcaption}
\usepackage{amssymb}
\usepackage{amsfonts}
\usepackage{amsmath}

\def\makeheadbox{{%
	}}

\newcommand{\RN}[1]{%
	\textup{\uppercase\expandafter{\romannumeral#1}}%
}

\def\bp{{\bar\partial}}
\def\bfs{\boldsymbol}
\newcommand{\ms}{\medskip}
\def\pa{\partial}
\def\sm{\setminus}

\def\wh{\widehat}
\def\ve{\varepsilon}
\def\Re{ \mathrm{Re}}
\def\Im{ \mathrm{Im}}

\def\C{\mathbb{C}}

\def\E{\mathbf{E}}

\def\R{\mathbb{R}}

\begin{document}

\title{The high temperature crossover for general 2D Coulomb gases
\thanks{Financial support to Gernot Akemann by the German Research Foundation (DFG) through CRC1283 ``Taming uncertainty and profiting from randomness and low regularity in analysis, stochastics and their applications" and to Sung-Soo Byun by Samsung Science and Technology Foundation (SSTFBA1401-01) are acknowledged.
Both authors are equally grateful to the DFG's International Research Training Group IRTG 2235 supporting the Bielefeld-Seoul exchange programme.}
}


\author{Gernot Akemann         \and
       Sung-Soo Byun 
}

\institute{Gernot Akemann   \at
            Faculty of Physics, Bielefeld University, P.O. Box 100131, 33501 Bielefeld, Germany \\
              \email{akemann@physik.uni-bielefeld.de}           
           \and
           Sung-Soo Byun \at
            Department of Mathematical Sciences, Seoul National University,\newline Seoul, 151-747, Republic of Korea
            \\
             \email{sungsoobyun@snu.ac.kr}   
}

\date{}

\maketitle

\begin{abstract}
We consider $N$ particles in the plane influenced by a general external potential that are subject to the Coulomb interaction in two dimensions at inverse temperature $\beta$. At large temperature, when scaling $\beta=2c/N$ with some fixed constant $c>0$, in the large-$N$ limit we observe a crossover from Ginibre's circular law or its generalization to the density of non-interacting particles at $\beta=0$. Using several different 
methods we derive a partial differential equation of generalized Liouville type for the crossover density. For radially symmetric potentials we present some 
asymptotic results and give examples for the numerical solution of the crossover density. 
These findings generalise previous results when the interacting particles are confined to the real line. In that situation we derive an integral equation for the resolvent valid for a general potential and present the analytic solution for the density in case of a Gaussian plus logarithmic potential.
\end{abstract}

\section{Introduction and Main Results}\label{Main}

Particle systems that interact logarithmically - the Coulomb repulsion in two dimensions (2D) -  and that are subject to a confining potential, at temperature $T$ parame\-trised by $\beta^{-1}=k_BT$, enjoy an intimate relationship with Random Matrix Theory, see e.g., \cite{forrester1998exact,forrester2010log}. Here, one has to distinguish two cases. 

When the particles are constrained to the real line or a subset of it, such systems can be realised as eigenvalues of random $N\times N$ matrices whose entries follow a  Gaussian or more general distribution. In that case, the inverse temperature 
$\beta$ takes the specific values 1, 2 and 4 for self-adjoined matrices with real, complex or quaternionic entries, and we refer to \cite{Mehta} for a discussion of these classical Gaussian ensembles. For general $\beta>0$ - the so-called $\beta$-ensembles - other realisations exist, such as tri-diagonal matrices \cite{dumitriu2002matrix} or in terms of Dyson's Brownian motion \cite{dyson1962brownian}, see also \cite{MR3078021} for an invariant realisation. While for large $N$ on a global scale, the limiting spectral density is given by Wigner's semi-circle for all $\beta>0$ for ensembles with Gaussian potential, on a local scale the statistics strongly depends on $\beta$. For the classical ensembles the local statistics of particles (or eigenvalues) is very well understood and known to be universal, (see e.g., \cite[Chapter6]{akemann2011oxford}), whereas progress for $\beta$-ensembles has been rather recent. It is given in terms of different stochastic differential operators in the bulk and at the edges, and we refer to \cite{ramirez2011beta,valko2009continuum,valko2017sine}.

Turning to the case when the particles move in the plane, thus representing a true 2D Coulomb gas, much less is known for general $\beta>0$. First, only for matrices with complex Gaussian entries without further symmetry - the complex Ginibre ensemble - the corresponding complex eigenvalues yield a Coulomb gas at $\beta=2$. For real or quaternionic matrix entries one obtains point processes of Pfaffian type \cite{edelman1997probability,lehmann1991eigenvalue,ginibre1965statistical} that differ from the standard 2D Coulomb gas at $\beta=1$ or 4\footnote{For a different interpretation of these real and quaternionic Ginibre ensembles as a multi component Coulomb gas we refer to \cite{forrester2016analogies}. }. Only normal random matrices with complex or quaternionic entries provide realisations at $\beta=2$ and $4$, see e.g., \cite{chau1998structure} and \cite{hastings2001eigenvalue}, respectively. The eigenvalue statistics of complex normal and complex Ginibre matrices happen to agree, but not their eigenvector statistics \cite{chalker1998eigenvector}. Again, on a global scale the limiting spectral density is given by Girko's circular law for all $\beta>0$ for a Gaussian potential. 
Relatively little is known about the local statistics beyond $\beta=2$. Only at the particular value $\beta=2$ the point process is determinantal, and local universality has been shown for invariant (see e.g.,\cite{GT10,hedenmalm2017planar,AKM14,akemann2016universality,berman2008determinantal}) and Wigner ensembles \cite{tao2015random}. For general $\beta$ the low temperature limit corresponding to $\beta\gg1$ is subject of on going research (see e.g.,\cite{Ameur18low}), due to the conjectured condensation on the so-called Abrikosov lattice, and we refer to \cite{serfaty2017microscopic} for a recent review and references. 

Recently, the opposite high temperature limit $\beta\to0$ has been studied for $\beta$-ensembles with real \cite{allez2012invariant,MR3078021,duy2015mean} or real positive eigenvalues \cite{allez2012invariant2}. Here, $\beta$ is not kept fixed in the large-$N$ limit and a different scaling $\beta=\frac{2c}{N}$ with a constant $c\in(-1,\infty)$, was identified in \cite{allez2012invariant}. There, the solution for the limiting global density $\rho_c(x)$ was given in terms of parabolic cylinder functions and was shown to interpolate between Wigner's semi-circle distribution at large $c\gg1$ and a Gaussian one at $c=0$. Furthermore, allowing for a weakly attracting interaction $c<0$, it is believed to converge towards a Dirac delta when $c\downarrow -1$. In this article we will study the corresponding limit for genuine 2D Coulomb gases in the plane, with a general confining potential. The possibility taking of such a limit, leading to a crossover between the circular law and a Gaussian density for a Gaussian potential, was already mentioned in \cite{chafai2016concentration,bolley2017dynamics}.  We will find an extended parameter range  $c\in(-2,\infty)$, with convergence to a Dirac delta  when $c\downarrow -2$. The latter was already observed and in fact proven for a Gaussian plus linear potential in \cite{caglioti1992special,caglioti1995special}. There, the limiting behaviour of so-called vortex systems in the plane was analysed and the existence of a solution for the limiting global density was shown. 

Before giving more details and presenting our results, let us briefly comment on our methods. Our approach will be threefold, combining  rigorous and heuristic methods. First, we start by representing our particle system in the plane by the stationary solution of a 2D diffusion process. Assuming its well-posedness for suitably chosen potentials,  we use  It\^o's calculus to derive an integral equation that relates the 1- and 2-point correlation functions, see Theorem \ref{Ward} below. The same theorem follows from our second method, the so-called loop equation or Ward identity, with less assumptions. These two methods have the advantage of being exact at finite-$N$, thus serving as a  starting point for both global and local analysis, cf. \cite{zabrodin2006large} for an expansion of the free energy and 2-point correlation function at $\beta>0$ in $1/N$. What we are currently lacking is a precise estimate for the factorisation of the 2-point function for general $\beta$.
Therefore, we will use a third heuristic method, the saddle point or variational approach (also called large deviations) to derive a mean field equation for the limiting global density. This approach has the advantage of making transparent, which terms contribute in which large-$N$ limit. On the one hand, keeping $\beta>0$ fixed always leads to the circular law or its generalisation, whereas scaling $\beta=\frac{2c}{N}$ leads to an interpolation between the circular law and the Gaussian distribution or their generalisations. Our methods can to large extent be pursued in parallel for particles on the real line or in the plane. This allows us to slightly generalise previous results \cite{allez2012invariant} on the line to general potentials, for which we will give an example.

Let us formulate our main results. In this section we will focus only on particles in the plane, representing a true 2D Coulomb gas. 
We study an ensemble of $N$ charged particles, that interact logarithmically under the influence of an external confining potential $Q$. 
Labelling the particle's positions by $\bfs{\zeta}=(\zeta_1,
\cdots,\zeta_{N}) \in \C^N$, the associated Gibbs measure at inverse temperature $\beta$ is given by 
\begin{align} \label{2D Cou m}
\displaystyle d\mathbf{P}_N(\bfs{\zeta})=  
p_N (\bfs{\zeta})\prod_{j=1}^{N}dA(\zeta_j)\ ,\ \ \ 
p_N (\bfs{\zeta})=  
\frac{1}{ Z_N}  \prod_{j>k=1}^N |\zeta_j-\zeta_k|^{\beta} e^{-m\sum_{j=1}^N Q(\zeta_j)} .
\end{align}
Here, $dA$ is the area measure (i.e., 2-dimensional Lebesgue measure divided by $\pi$), $p_N (\bfs{\zeta})$ is the joint density of particles, and $Z_N$ stands for the normalising partition function. 
The choice of the scaling parameter $m$, that may depend on $N$ and $\beta$, determines the limiting behaviour of our ensemble. In order to distinguish its r\^ole from 
$N$, several authors identify it with the inverse Planck constant $1/\hbar$, see e.g., \cite{zabrodin2006large} \footnote{Note that these and several other authors \cite{AKM14} use a different convention, denoting $\beta= 2\tilde{\beta}$.}.
Throughout this article, we assume that $Q$ is smooth and sufficiently large near the infinity (e.g., $Q(z) \gg \log |z|$) so that $Z_N < \infty$.

The quantities determining the system $\bfs{\zeta}=\{ \zeta_j\}^{n}_{j=1}$ are the following $k$-point correlation functions defined as the expectation values $\mathbf{E}_N$ with respect to the Gibbs measure \eqref{2D Cou m}:  
\begin{equation} \label{k-pt cor}
R_{N,k}(z_1,\cdots,z_{k}):= N^k\ \mathbf{E}_N\left( \prod_{l=1}^N\rho_N(z_j)\right)\ , \quad \text{for}\quad k=1,\cdots, N\ ,
\end{equation}
when all arguments are mutually distinct, $z_i\neq z_j$, $\forall i,j=1,\ldots,k$, and zero for any pair of arguments coinciding.
Here, 
\begin{equation}
\label{norm-count}
\rho_N(z)=\frac1N \sum_{j=1}^N \delta(z-\zeta_j)
\end{equation}
is the normalised counting function.
We remark that once properly normalised, the $R_{N,k}(z_1,\ldots,z_{k})$ can be interpreted as the probability to find $k$ particles at given positions $z_1,\ldots,z_{k}$.

We begin with our first method, It\^{o}'s stochastic calculus. First, we observe that $p_N(\zeta)$ in \eqref{2D Cou m} is the stationary solution of the following 2D diffusion process
\begin{equation} \label{diff QzIntro}
d\zeta_j(t)=\sqrt{2} \, dz_j(t)-2m \, \bp Q(\zeta_j(t)) dt +\beta \sum_{k: k\not= j}^N \frac{1}{\overline{\zeta_j}(t)-\overline{\zeta_k}(t)}dt\ ,
\end{equation}
where $z_j(t)$ a standard 2D Brownian motion. 
In the case of the Gaussian potential, the well-posedness of such a system was shown by Bolley, Chafa\"{i} and Fontbona, see \cite{bolley2017dynamics}. We refer the reader to \cite[Section 4.3]{anderson2010introduction} and references therein for some basic properties of such dynamical systems. 
Applying It\^o's lemma for complex variables, we can then prove the following theorem.
In fact we state a version that follows from the Ward identities as shown in Section~\ref{WARD Sec}, with weaker assumptions on the confining potential $Q$.

\begin{theorem}\label{Ward}
Given Gibbs measure \eqref{2D Cou m} with a $C^2$-smooth potential $Q$,
the following relation between 1- and 2-point correlation functions holds for every finite $N$: 
\begin{equation} \label{LOOP}
\frac{\beta}{2} \int_{\C} \frac{R_{N,2}(\zeta,\eta)  }{\zeta-\eta} dA(\eta)= m(\pa_\zeta Q(\zeta)) R_{N,1}(\zeta) + \pa_\zeta R_{N,1}(\zeta).
\end{equation}
\end{theorem}

Equation \eqref{LOOP} can be used as a starting point for a systematic expansion in the large-$N$ limit, 
cf. \cite{zabrodin2006large} for earlier work.
Let us introduce the connected 2-point correlation function
\begin{equation}\label{R2conn}
R_{N,2}^{\rm conn}(\zeta,\eta) := R_{N,2}(\zeta,\eta) -R_{N,1}(\zeta) R_{N,1}(\eta)\ . 
\end{equation}
For a nonvanishing $R_{N,1}(\zeta) \neq0$ we can then rewrite eq. \eqref{LOOP} as follows
\begin{equation} \label{LOOP2}
\frac{\beta}{2} \int_{\C} \frac{R_{N,1}(\eta)  -B_N(\zeta,\eta)}{\zeta-\eta} dA(\eta)
= m\pa_\zeta Q(\zeta)+ \pa_\zeta \log[R_{N,1}(\zeta)].
\end{equation}
Here, $B_{N}(\zeta,\eta) :=-{R_{N,2}^{\rm conn}(\zeta,\eta) }/{R_{N,1}(\zeta)}$ is defined such that it corresponds to the Berezin-kernel at $\beta=2$. While this is a well-studied object at $\beta=2$, little is known for general $\beta>0$, see however some remarks in \cite{AKM14}. In order to arrive at the mean field equation \eqref{MFE} below, that determines the limiting density in the particular large-$N$ limit that we consider, we would have to show that the connected 2-point function \eqref{R2conn} is sub-leading. This is equivalent to show the factorisation of the 2-point function on the global level - a property called mean field or propagation of chaos - and we expect it to hold up to order $O(N^{-2})$, cf. \cite{valko2017sine}.

Let us turn to the detailed analysis of the global large-$N$ behaviour of \eqref{2D Cou m} in the high temperature regime $\beta \rightarrow 0$.
It is clear that this regime implies weaker correlation among particles. In the extreme case $\beta\equiv0$, the particles become independent from each other, their $k$-point correlation functions trivially factorise and become proportional to $\prod_{j=1}^k e^{-Q(\zeta_j)}$, normalised with respect to the area measure. 
Our main purpose in this paper is to investigate the crossover phenomenon between fixed and vanishing $\beta$. 
For instance, in the case of a Gaussian potential $Q(\zeta)=|\zeta|^2$, we study the smooth interpolation between Ginibre's circular law and the Gaussian distribution. 
The possibility of such a crossover regime was already mentioned in \cite{bolley2017dynamics}. 
The precise scaling we have to impose in \eqref{2D Cou m} is to set 
\begin{equation} \label{beta n 2c}
m=1\quad \text{and} \quad \beta=\frac{2c}{N} \quad \text{for fixed} \quad c\in (-2,\infty).
\end{equation} 
Here, $c$ is kept fixed when $N\to\infty$, and we can allow for a weakly attracting interaction with negative $c$ as well. 
The same scaling \eqref{beta n 2c} was found on the real line in one dimension \cite{MR3078021}, however with fixed $c\in (-1,\infty)$.

On the other hand, for the more standard scaling 
\begin{equation} \label{m scale1}
m=\frac{\beta}{2} \, N, \quad \text{with fixed} \quad \beta>0\ ,
\end{equation}  
Chafa\"{i}, Hardy and Ma\"{i}da showed that if $\beta$ is bigger than $\beta_0 \log N/N$ for some constant $\beta_0$, there is no such crossover phenomenon
and the limiting global density follows \eqref{Laplacian growth} below, see 
\cite{chafai2016concentration}.

In Section~\ref{Saddlepoint}, we will heuristically calculate the \textit{free energy functional} $F[\rho]$ in terms of the probability density function $\rho$, that is  associated to the Gibbs measure \eqref{2D Cou m}. 
Here, we will utilize the saddle point method in the large-$N$ limit. In the  high temperature regime \eqref{beta n 2c} we obtain the following formula 
\begin{align} \label{free}
\begin{split}
F\equiv F_c[\rho]=&  \int_{\C} Q(\zeta) \rho(\zeta) dA(\zeta)-c \int_{\C^2} \log|\zeta-\eta| \rho(\zeta) \rho(\eta) dA(\zeta) dA(\eta)
\\
&+\int_{\C} \rho(\zeta) \log \rho(\zeta) dA(\zeta).
\end{split}  
\end{align} 
While the first line can be easily seen to follow from the energy in \eqref{2D Cou m}, the second line is the so-called entropy contribution.
The saddle point condition 
\begin{equation}
\label{SPeq}
0=\frac{\delta F_c[\rho] }{\delta\rho(\zeta)}=Q(\zeta) -2c \int_{\C} \log|\zeta-\eta| \rho(\eta)dA(\eta)+\log \rho(\zeta)+1\  
\end{equation}
is imposed in order to extremise the free energy. Equation \eqref{SPeq}  
has the limiting density $\rho_c(\zeta)=\lim_{N\to\infty}\frac1N R_{N,1}(\zeta)$ as its solution. Applying the Laplace operator 
$\Delta=\pa\bp$ to \eqref{SPeq} and using that its Green's function is the logarithm, we obtain that the crossover density $\rho_c$ satisfies \eqref{MFE} below.
This leads us to propose the following extension of \cite[Theorem 6.1]{caglioti1992special} for a general potential.
\begin{theorem}\label{Thmfe} 
The limiting density function $\rho_c$ minimises (resp., maximises) the free energy $F_c[\rho]$ for $c>0$, (resp., $<0$) and solves the following  \textit{mean field equation}: 
\begin{equation} \label{MFE}
0 = \Delta Q(\zeta)-c \, \rho_c(\zeta) + \Delta \log \rho_c(\zeta)\ .
\end{equation}
\end{theorem}

We wish to emphasize that we currently do not have a complete proof for this statement. However, 
if the factorisation or mean field property of the $2$-point correlation function \eqref{R2conn} holds, in the sense that for any continuous, bounded function $f(\zeta,\zeta')$ on $\C^2$, 
\begin{align}
\begin{split}
\label{fact}
&\lim_{N \rightarrow \infty} \frac{1}{N^2}\int_{\C^2} f(\zeta,\zeta') R_{N,2}(\zeta,\zeta') dA(\zeta) dA(\zeta') 
\\
=&  
\int_{\C^2} f(\zeta,\zeta')\rho_c(\zeta) \rho_c(\zeta') dA(\zeta) dA(\zeta'),
\end{split}
\end{align} 
the mean field equation \eqref{MFE} follows from \eqref{LOOP} in Theorem \ref{Ward}.
Namely, imposing the scaling \eqref{beta n 2c} on \eqref{LOOP} and normalising the 1-point function $R_{N,1}(\zeta)$ by $ \frac1N $, the anti-holomorphic
derivative $\bp_{\zeta}$  of the limit of \eqref{LOOP} directly leads to \eqref{free}, when neglecting the contribution from the limit of $B_N$ in the sense of \eqref{fact}. For the minimising (maximising) property we only have heuristic arguments.

\begin{remark}
For the choice of potential 
\begin{equation}
\label{Qvortex}
Q(z) \equiv Q_N(z):=\lambda|z|^2+\frac{1}{2N} \left( \eta \bar{z} + \bar{\eta} z  \right),  \quad\text{with fixed} \quad  \lambda>0,\ \eta \in \C, 
\end{equation}
the joint distribution \eqref{2D Cou m} can be identified with the system of stationary states of $N$ vortices in the plane, cf \cite{caglioti1992special}. In \cite[Theorem 6.1]{caglioti1992special} in the limit \eqref{beta n 2c} (using different conventions for our constant $c$) the free energy functional of vortices \eqref{free} was rigorously derived for potential \eqref{Qvortex}.  
The mean field equation \eqref{MFE} for this potential was proven, including the
existence and extremising properties of its solution. The authors also showed  convergence towards the Dirac measure in the limit $c\downarrow-2$, see \cite{caglioti1995special}.
\end{remark}

We now compare the above free energy \eqref{free} and resulting mean field equation \eqref{MFE} to the standard large-$N$ scaling limit \eqref{m scale1}, which is well understood. Here, 
only the first line in \eqref{free} will contribute in this limit, leading to the weighted-logarithmic energy functional (see \cite{ST97}) 
\begin{equation} \label{Frost Hamil}
F[\rho]=  \int_{\C} Q(\zeta) \rho(\zeta) dA(\zeta)-\int_{\C^2} \log|\zeta-\eta| \rho(\zeta) \rho(\eta) dA(\zeta) dA(\eta)\ .
\end{equation} 
Indeed, it was shown by Hedenmaln and Makarov that under some regularity and growth assumptions on $Q$, the one particle distribution $\rho$ weakly converges toward the equilibrium measure minimising $F[\rho]$, see \cite{HM13}. 
Moreover, by standard logarithmic potential theory (see \cite{ST97}), 
\begin{equation} \label{Laplacian growth}
0= \Delta Q(\zeta) -\rho(\zeta) 
\end{equation}
is valid on the limiting support of the measure $S$ which is called the \textit{droplet}. 
For Gaussian potential $Q(z)=|z|^2$ for example, this gives the circular law, with a constant density  on the unit disc. Note that \eqref{Laplacian growth} also can be obtained from \eqref{Frost Hamil} by requiring a saddle point condition as  in \eqref{SPeq}.
We emphasize that the standard choice of scaling \eqref{m scale1} makes the droplet and density $\rho$ independent of the inverse temperature $\beta$. 

We return to the discussion of the mean field equation \eqref{MFE}. Defining $\psi_c := \log \rho_c$, it is rewritten as follows 
\begin{equation} \label{quasi Liouville}
c \, e^{\psi_c(\zeta)}  = \Delta Q(\zeta)+\Delta \psi_c(\zeta) ,
\end{equation}
which is a differential equation of \textit{generalised Liouville type}.
In case that $\Delta Q \equiv 0$ would hold, the equation \eqref{quasi Liouville} reduces to the standard Liouville equation whose explicit solutions are well-known, see e.g., \cite{crowdy1997general}. However, also in view of the result \eqref{Laplacian growth} in the standard scaling limit, we cannot assume that $\Delta Q$ is small or even negligible in any sense. For that reason we have been unable to provide an explicit solution for \eqref{MFE}, or equivalently \eqref{quasi Liouville}, even in the Gaussian case.  
We are unaware of a deeper relation between Liouville's equation and Dyson's Brownian motion in general in 2D. However, let us mention \cite{david2015renormalizability} where methods from Gaussian multiplicative chaos were utilized in the renormalisation of Liouville quantum gravity.

Let us discuss now several special cases.
For the choice of a radially symmetric potential we can provide the asymptotic behaviour of the limiting density for large $r=|z|$.
In this case we can explicitly check the interpolating property of the solution to \eqref{MFE} in the limits $c\to0$ and $c\gg1$, as we will further exemplify for monic so-called Freud potentials, that are a special case of Mittag-Leffler potentials named in \cite{ameur2018random}. In addition, we will present two examples for a numerical solution of \eqref{MFE}, for a Gaussian and quartic monic potential.

\ms
\noindent $\rhd$ \textbf{Radially symmetric potentials.} 
Suppose that the external potential is radially symmetric, i.e., there exists a function $f: \R_+ \rightarrow \R$ satisfying
\begin{equation}
f(r) = f(\sqrt{x^2+y^2}) := Q(z), \quad \text{with}\quad  z=x+iy.
\end{equation}
Let us denote by
\begin{equation} \label{radial exp gen}
g_c(r)= g_c(\sqrt{x^2+y^2}):= \frac{1}{\pi}\,\rho_c(z)
\end{equation}
the radial part of crossover density $\rho_c$. Here, the factor $1/\pi$ comes from the fact that $\rho_c(z)$ is a density function with respect to the area measure. By definition, we have 
\begin{equation} \label{int rden gen}
\int_{0}^{\infty}  r \, g_c(r) dr=\frac{1}{2\pi} 
\end{equation}
for the normalisation.
Notice that the 2D Laplace operator $\Delta$ acts on the radial density as 
\begin{equation} \label{rad Lap}
\Delta =\frac{1}{4}\left(\pa_x^2+\pa_y^2\right)= \frac{1}{4} \left( \pa_r^2+\frac{1}{r} \pa_r  \right).
\end{equation}
Combining \eqref{MFE} and \eqref{rad Lap}, we obtain following ordinary differential equation for the radial crossover density:
\begin{equation} \label{r ODE gen}
4\pi c \, g_c(r)=f''(r)+\frac{1}{r} \, f'(r) + \frac{g''_c(r)g_c(r)-\left(g'_c(r) \right)^2}{g_c(r)^2}+\frac{1}{r}\,\frac{ g'_c(r)}{g_c(r)}\ ,
\end{equation}
where we have put the density on the left-hand side. 
The asymptotic behaviour of $g_c$ for large radii,
\begin{equation}
\label{large-r}
g_c(r)=r^{2c}e^{-f(r)+o(\log r)}, \quad \text{as} \quad r\rightarrow \infty\ ,
\end{equation}
can be easily seen. Multiplying \eqref{r ODE gen} by $r$ and integrating it using the normalisation \eqref{int rden gen}, we obtain
\begin{equation}
\label{ODE-int}
2c=\left[rf'(r)+r\left(\log g_c(r)\right)'\right]_0^\infty\ ,
\end{equation}
from which \eqref{large-r} follows. In fact the function $r^{2c}e^{-f(r)}$ solves the ``homogeneous" equation \eqref{r ODE gen}, where the left-hand side is set to zero. However, due to the non-linearity of the equation, the solution is not given by this ``homogeneous" solution plus a special solution.

\ms
\noindent $\rhd$ \textbf{Examples.}
A particular realisation of a rotationally invariant potential is given by  the monomials, so-called \textit{Freud} or \textit{Mittag-Leffler potentials} (cf. \cite{ameur2018random})
\begin{equation}
Q(z)=  \frac12  |z|^{2\alpha}, \quad \alpha \ge 1.
\end{equation}
In this case we obtain for $r$ times \eqref{r ODE gen} 
\begin{equation} 
\label{radial-Freud}
4\pi c \, r \, g_c=2\alpha^2 r^{2\alpha-1} +\left(r \, (\log g_c)'\right)'.
\end{equation}
Note that for these homogeneous potentials, the ensemble \eqref{2D Cou m} with scales $m=\beta N/2$ and $m=1$ can be related by simple rescaling of the point particles. Therefore, by \eqref{Laplacian growth}, it is easy to calculate the radial density in the limit when $c \rightarrow \infty$. As a result, the extremal cases of the solutions of \eqref{MFE} including their normalisations are given by 
\begin{equation} 
\label{c-limits}
g_c(r)\sim \begin{cases}
\displaystyle \frac{1}{\pi\, 2^{1/\alpha}\, \Gamma(1+1/\alpha)} \, \exp\left(-\frac12\, r^{2\alpha}\right) &\text{as} \quad c \rightarrow 0;
\vspace{0.5em}
\\
\displaystyle \frac{\alpha^2}{2\pi c} \, r^{2\alpha-2}\, \mathbbm{1}_{[0,(2c/\alpha)^{1/2\alpha}]}&\text{as} \quad c \rightarrow \infty.
\end{cases}
\end{equation}
These are of course just special cases for $g_c(r)\sim e^{-f(r)}$ for $c\to0$ and $g_c(r)\sim \Delta f(r) \mathbbm{1}_S$ for $c\to\infty$ on the corresponding droplet $S$.

We present now two examples for numerical solutions of \eqref{radial-Freud} at specific values of $c$. In Figure~\ref{Ginibre fig} below the case $\alpha=1$ of a Gaussian and in Figure~\ref{Freud fig} of a quartic potential with $\alpha=2$ are shown.
We obtain the numerical solutions not only for positive $c$ but also for negative $c$. The conjecture is that as $c$ goes to its critical (negative) value $-2$, the ensemble collapses at the origin, i.e., the one particle density converges towards a Dirac delta. While for $\alpha=1$ this is known \cite{caglioti1992special} we observe that a similar behaviour occurs for $\alpha=2$.


\newpage

\begin{figure}
	\begin{center}
		\begin{subfigure}{0.495\textwidth}
			\includegraphics[width=2.5in]{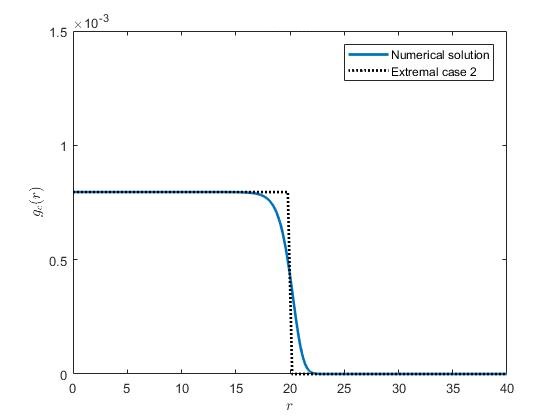}
			\caption{$c=200$}
		\end{subfigure}
		\hfill
		\begin{subfigure}{0.495\textwidth}
			\includegraphics[width=2.5in]{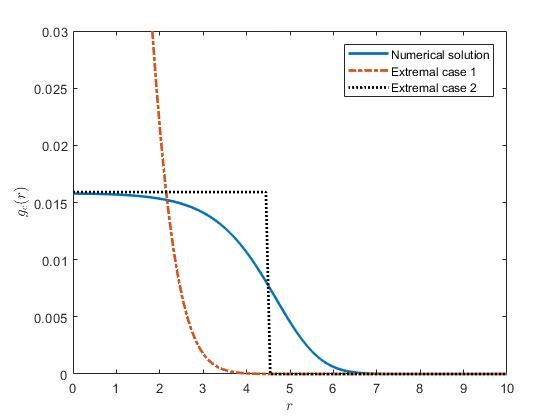}
			\caption{$c=10$}
		\end{subfigure}
		\vfill 
		\begin{subfigure}{0.495\textwidth}
			\includegraphics[width=2.5in]{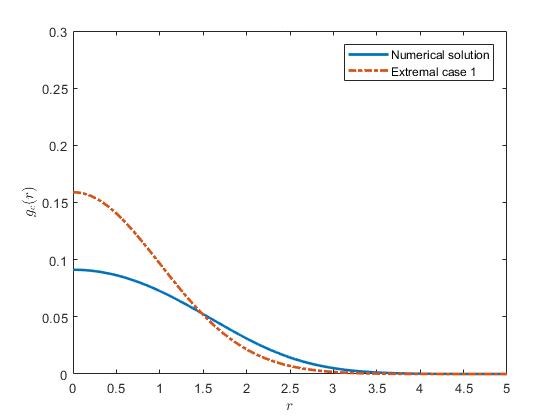}
			\caption{$c=1$}
		\end{subfigure}
		\hfill
		\begin{subfigure}{0.495\textwidth}
			\includegraphics[width=2.5in]{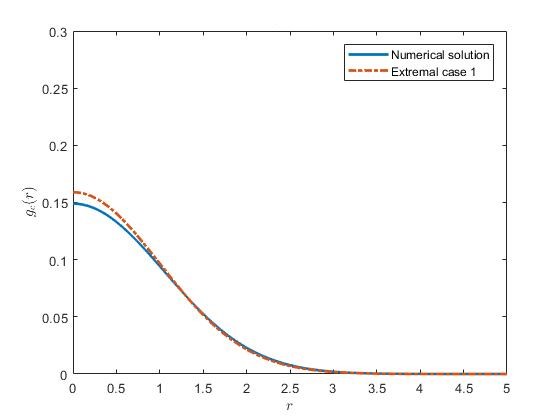}
			\caption{$c=0.1$}
		\end{subfigure}
		\vfill
		\begin{subfigure}{0.495\textwidth}
			\includegraphics[width=2.5in]{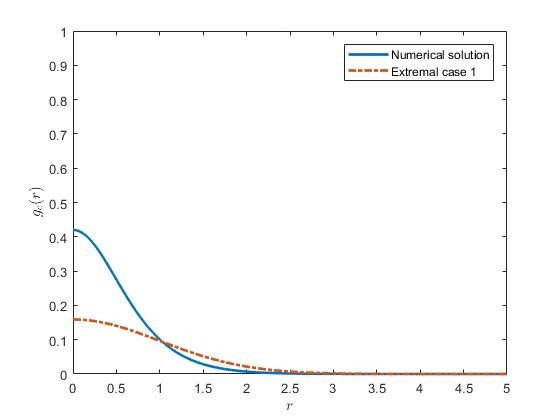}
			\caption{$c=-1$}
		\end{subfigure}
		\hfill
		\begin{subfigure}{0.495\textwidth}
			\includegraphics[width=2.5in]{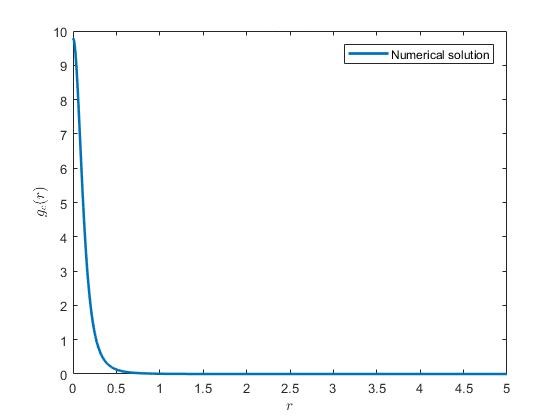}
			\caption{$c=-1.9$}
		\end{subfigure}
		\caption{The numerical solution for our interpolating radial density $g_c(r)$ (full blue line) of our mean field equation \eqref{radial-Freud} is shown for the Gaussian potential		
		$Q(z)=|z|^2/2$, with parameter $c$ decreasing from $c=200$ to $c=-1.9$. What is also shown is the limiting circular law from \eqref{c-limits} for large $c$ (two top plots, dotted lines) as well as the limiting Gaussian density (dashed orange lines) for $c=0$. Note that the normalisation is with respect to the radial measure $2\pi\,r$ in 2D, see \eqref{int rden gen}. For that reason the area under the curves does not agree.}  
		\label{Ginibre fig}
		\vfill
	\end{center}
\end{figure}

\clearpage

\begin{figure}
	\begin{center}
		\begin{subfigure}{0.495\textwidth}
			\includegraphics[width=2.5in]{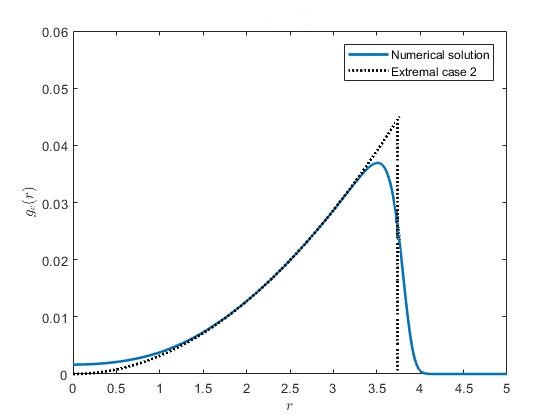}
			\caption{$c=200$}
		\end{subfigure}
		\hfill
		\begin{subfigure}{0.495\textwidth}
			\includegraphics[width=2.5in]{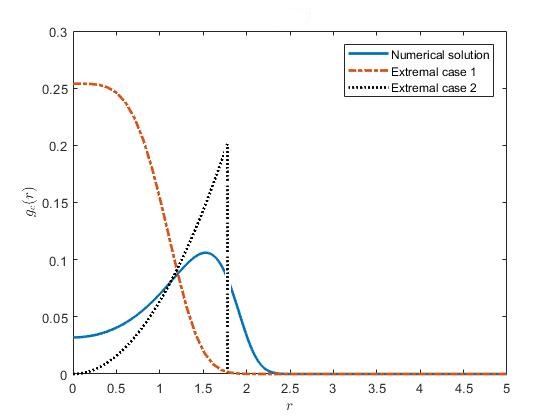}
			\caption{$c=10$}
		\end{subfigure}
		\vfill 
		\begin{subfigure}{0.495\textwidth}
			\includegraphics[width=2.5in]{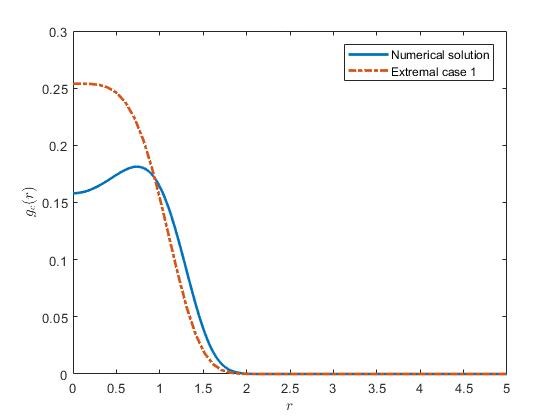}
			\caption{$c=1$}
		\end{subfigure}
		\hfill
		\begin{subfigure}{0.495\textwidth}
			\includegraphics[width=2.5in]{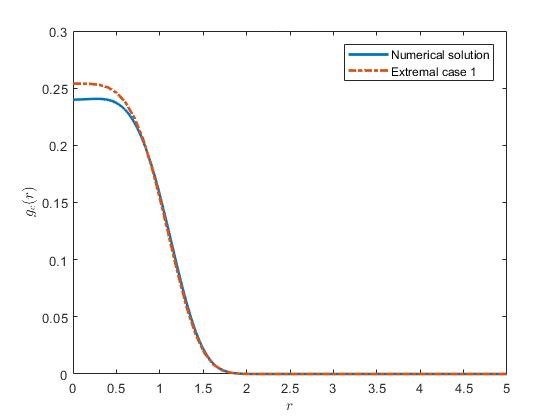}
			\caption{$c=0.1$}
		\end{subfigure}
		\vfill
		\begin{subfigure}{0.495\textwidth}
			\includegraphics[width=2.5in]{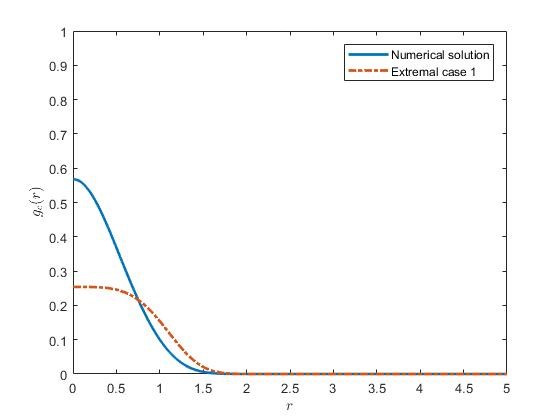}
			\caption{$c=-1$}
		\end{subfigure}
		\hfill
		\begin{subfigure}{0.495\textwidth}
			\includegraphics[width=2.5in]{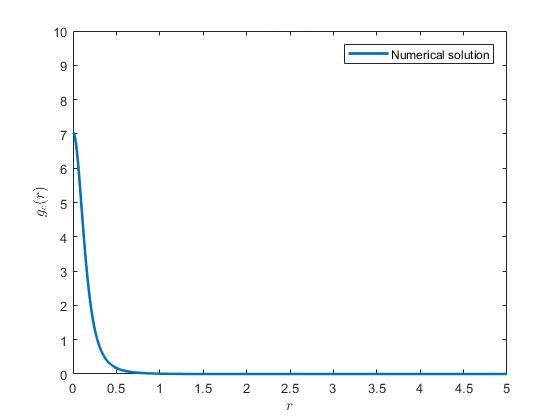}
			\caption{$c=-1.9$}
		\end{subfigure}
		\caption{The same plots as in Figure \ref{Ginibre fig} are shown for a quartic potential $Q(z)=|z|^4/2$. Here, for large $c$ the limiting circular law is replaced by a parabola, and for $c=0$ the Gaussian by $e^{-r^4/2}$, see \eqref{c-limits}. The approach to the conjectured Dirac delta looks similar to the previous figure at $c=-1.9$.}\label{Freud fig}
			\vfill
	\end{center}
\end{figure}

\clearpage

The remainder of this article is organised as follows. In Section~\ref{Stochastic}, we will approach our Coulomb gas in 2D and also in 1D as a diffusive process. Here a first version of Theorem~\ref{Ward} will be proven, including its 1D counterpart.  Section~\ref{WARD Sec} is devoted to the study of the Ward identity and the final version of Theorem~\ref{Ward}. Its 1D version follows in parallel, and the corresponding free energies result when assuming factorisation. In Section~\ref{Saddlepoint}, we will introduce the saddle point method in a heuristic way. Here, the two different scalings \eqref{beta n 2c} and \eqref{m scale1} leading to the respective free energies \eqref{free} and \eqref{Frost Hamil} will become evident. This leads to the mean field equations given above. In addition, in the 1D case we derive a mean field equation for the resolvent for a general potential, slightly generalising \cite{allez2012invariant}. We then give an example for a Gaussian plus logarithmic potential
\begin{equation}
V_a(x)= \frac12 x^2-a\log|x|\ , \quad a>-1.
\end{equation}
The associated resolvent equation can be solved, following \cite{allez2012invariant2} closely. The resulting interpolating density $\rho_c(\lambda)$ is then given by Kummer's (confluent) hypergeometric function, see \eqref{Gausslog}, with special cases shown in Figures~\ref{Gausslog1}, \ref{Gausslog0} and \ref{Gausslog2}.

\begin{figure}[h]
	\begin{center}
		\includegraphics[width=4.3in]{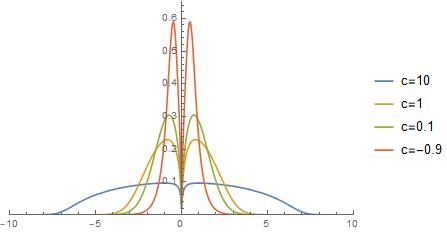} 	\caption{The plot displays the interpolating densities $\rho_c(\lambda)$ against $\lambda\in\R$. It is given in \eqref{Gausslog} for $a=\frac12$ with parameter $c$ decreasing from $c=10$ to $c=-0.9$. The figure indicates the logarithmic repulsion at the origin. For large $c$ the density converges to the semi-circle and this repulsion only becomes visible on a local scale.
		Here and below in Fig. \ref{Gausslog2} the normalisation constant of the density is determined numerically.
		The convergence to the conjectured Dirac delta at $c=-1$ is also visible.} \label{Gausslog1}
	\end{center}
\end{figure}	
	\begin{figure}
	\begin{center}
		\includegraphics[width=4.3in]{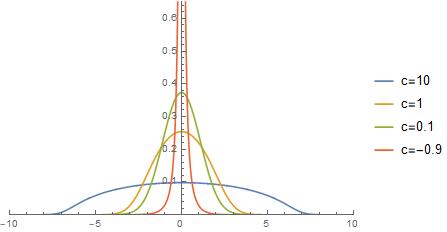}  \caption{The same plot as in Figure~\ref{Gausslog1} for $a=0$, reproducing the findings of \cite{allez2012invariant}.
		It shows the progressive transition from Wigner's semi-circle ($c=\infty$) through the Gaussian distribution ($c=0$) to the conjectured Dirac delta ($c=-1$). 
		}  \label{Gausslog0}
	\end{center}
	\end{figure}	
	\begin{figure}
	\begin{center}
		\includegraphics[width=4.3in]{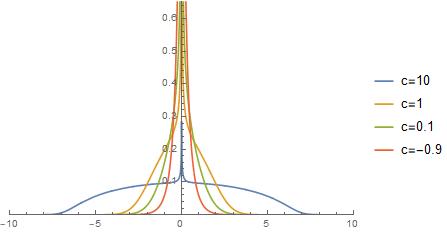} 	\caption{The same plot as in Fig. \ref{Gausslog1} for $a=-\frac12$, representing a logarithmic attraction towards the origin. } \label{Gausslog2}
	\end{center}
\end{figure}

\vfill


\section{Stochastic Dynamics for Coulomb Gases in 2D and 1D} \label{Stochastic}

\subsection{Dynamics for 2D Coulomb gases} \label{ITO Sec}
We begin by setting up the framework for the Dyson type dynamics whose invariant law is given by Gibbs measure \eqref{2D Cou m}.
For a given external potential $Q$, let us consider the 2D diffusion process
\begin{equation}
\zeta_j(t)=x_j(t)+iy_j(t), \quad x_j, y_j \in \R, \quad j=1,\cdots, N
\end{equation} 
where the 1D diffusion processes $x_j(t)$, $y_j(t)$ are given by 
\begin{align}
dx_j&=\sqrt{2} \, dB_j-m\pa_{x_j} Q \big( x_j+iy_j \big) dt+\beta \sum_{k: k\not= j}^N \frac{x_j-x_k}{(x_j-x_k)^2+(y_j-y_k)^2}dt, \label{diff Qx}
\\
dy_j&=\sqrt{2} \, d\tilde{B}_j-m\pa_{y_j} Q \big( x_j+iy_j \big)dt+\beta \sum_{k: k\not= j}^N \frac{y_j-y_k}{(x_j-x_k)^2+(y_j-y_k)^2}dt.  \label{diff Qy}
\end{align}
Here $B_j$ and  $\tilde{B}_j$ are independent 1D Brownian motions. 
Note that the system \eqref{diff Qx} and \eqref{diff Qy} can be rewritten as
\begin{equation} \label{diff Qz}
d\zeta_j(t)=\sqrt{2} \, dz_j(t)-2 \, m\bp Q(\zeta_j(t)) dt +\beta \sum_{k: k\not= j}^N \frac{1}{\overline{\zeta_j}(t)-\overline{\zeta_k}(t)}dt,
\end{equation}
where $z_j(t)$ is a 2D standard Brownian motion. For each time $t$, let $p_N(t,\bfs{\zeta})$ be the joint probability density function with respect to the area measure. Throughout this subsection we assume that the potential $Q$ is properly chosen so that the diffusive system of particles \eqref{diff Qz} is well-defined and admits a unique invariant measure. 
Under these assumptions, by virtue of standard It\^{o}'s calculus, one can easily show that the stationary density function $p_N(\bfs{\zeta}):=\lim_{t\rightarrow \infty} p_N(t,\bfs{\zeta})$ is given by
	\begin{align} \label{eq bet Q0}
	p_N(\bfs{\zeta})&=\frac{1}{Z_N}\prod_{j,k:j> k}^N |\zeta_j-\zeta_k|^\beta e^{-m\sum_{j=1}^N Q(\zeta_{j}) },
	\end{align}
	where $Z_N$ is the normalization constant. For being self-consistent, we present a sketch of the proof.

\begin{proof}
	Let $F: \R^{2N} \simeq \C^N \rightarrow \R$ be a given smooth function. From now on, we introduce a subscript for the corresponding differential operator acting on $F$, e.g., $\pa_{x_j} F= F_{x_j},  \pa_{x_j}^2 F=F_{x_jx_j}$. We write  $\E_t$ for the expectation with respect to $p_N(t,\bfs{\zeta})$, i.e., 
	\begin{equation*} 
	\E_t \, F(\bfs{\zeta})= \int F(\bfs{\zeta}) \, p_N(t,\bfs{\zeta}) \prod_{j=1}^N dA(\zeta_j).
	\end{equation*}  	
By It\^{o}'s formula, we have
\begin{align*} 
	\begin{split}
	dF(\bfs{\zeta})
	=&\sum_{j=1}^{N} \left(F_{x_j} dx_j+F_{y_j}dy_j\right)+\sum_{j=1}^N \left(F_{x_jx_j}+F_{y_jy_j} \right) dt.
	\end{split}
	\end{align*}
By \eqref{diff Qx}, \eqref{diff Qy}, we have 
	\begin{align*} 
	\begin{split}
	dF(\bfs{\zeta})
	=& \sum_{j=1}^{N} \sqrt{2} \left(  F_{x_j}dB_j+F_{y_j}d\tilde{B}_j \right)-m\left(\pa_{x_j} Q \cdot  F_{x_j}+\pa_{y_j} Q \cdot  F_{y_j} \right)dt 
	\\
	&+\beta \sum_{j,k:j\neq k}^N \left[ \frac{x_j-x_k}{(x_j-x_k)^2+(y_j-y_k)^2} F_{x_j}+ \frac{y_j-y_k}{(x_j-x_k)^2+(y_j-y_k)^2} F_{y_j} \right] dt
	\\
	&+\sum_{j=1}^N \left(F_{x_jx_j}+F_{y_jy_j} \right) dt.
	\end{split}
	\end{align*}
	Taking expectation on both sides, we obtain 
	$\pa_t \E_t \, F(\bfs{\zeta}) = \E_t \, D^\beta F( \bfs{\zeta}), $	where the operator $D^\beta$ acts on $F$ as 
	\begin{align*}
	D^\beta :=& \sum_{j=1}^{N} \left(\pa_{x_j}^2 +\pa_{y_j}^2- m\pa_{x_j} Q \cdot \pa_{x_j} -m\pa_{y_j} Q  \cdot \pa_{y_j}\right)
	\\
	&+\beta \sum_{j,k:j\neq k}^N  \frac{(x_j-x_k)\cdot \pa_{x_j}+(y_j-y_k)\cdot \pa_{y_j}}{(x_j-x_k)^2+(y_j-y_k)^2}. \nonumber
	\end{align*} 
	Using integration by part, we obtain that the stationary density function $p_N$ satisfies following partial differential equation:
	\begin{align*}
	0 =&\sum_{j=1}^{N} \frac{\left(\pa_{x_j}^2+\pa_{y_j}^2\right)p_N}{p_N} 
	+m\sum_{j=1}^{N} \left[  \left( \pa_{x_j}^2  Q + \pa_{y_j}^2 Q  \right)  + \left( \pa_{x_j} Q \, \frac{\pa_{x_j} p_N}{p_N} + \pa_{y_j} Q \, \frac{\pa_{y_j} p_N}{p_N} \right) \right] 
	\\
	&-\beta   \sum_{j,k:j\neq k}^N \left[\frac{x_j-x_k}{(x_j-x_k)^2+(y_j-y_k)^2} \frac{\pa_{x_j} p_N}{p_N}+ \frac{y_j-y_k}{(x_j-x_k)^2+(y_j-y_k)^2} \frac{\pa_{y_j} p_N}{p_N}  \right]. \nonumber
	\end{align*}
 Given our assumptions, all we need to check is that \eqref{eq bet Q0} solves this partial differential equation, which follows by direct calculations. $\Box$\\
\end{proof}

\subsubsection{Proof of Theorem \ref{Ward}}
Now we prove \eqref{LOOP}  by means of It\^{o}'s stochastic calculus. 
Let us consider the system of diffusion processes $\bfs{\zeta}(t)$ given by \eqref{diff Qz}, under the conditions on $Q$ that this is well-posed. Let $f$ be a (real-valued) smooth function defined on the complex plane. We define the time dependent normalised one-point counting function in the plane
\begin{equation*}
\rho_t(z):= \frac{1}{N} \sum_{j=1}^{N} \delta\big(z-\zeta_{j}(t) \big),
\end{equation*}
such that  
$$
\int f(z) \rho_t(z)dA(z) = \frac{1}{N}\sum_{j=1}^{N} f\big(\zeta_{j}(t)\big).
$$
By a complex-variable version of It\^{o}'s lemma, 
we obtain 
\begin{align*}
&d \int f(z) \rho_t(z)dA(z) \\
=& \frac{1}{N} \sum_{j=1}^{N}\Big[ \pa_{\zeta_j} f(\zeta_{j})d\zeta_{j}+\bp_{\zeta_{j}} f(\zeta_{j}) d\bar{\zeta}_j +4 \pa_{\zeta_j} \bp_{\zeta_j} f(\zeta_j) dt\Big] 
\\
= &  \frac{1}{N} \sum_{j=1}^{N}\Big[ \pa_{\zeta_{j}} f(\zeta_{j})d\zeta_{j}+\bp_{\zeta_{j}} f(\zeta_{j}) d\bar{\zeta}_j\Big] + 4 \left( \int (\pa \bp f) (z) \rho_t(z)dA(z) \right)dt.
\end{align*} 
Substituting $d\zeta_j$ by \eqref{diff Qz}, we have 
\begin{align*}
\begin{split} 
&d \int f(z) \rho_t(dz) 
- \frac{\sqrt{2}}{N} \sum_{j=1}^{N} \left( \pa f dz_j +\bp f d\bar{z}_j \right) 
\\
=&-\left[2 \int \left( \pa f m\bp Q -\pa \bp f  \right)\rho_t dA -N\beta \int_{z\neq w} \frac{\pa f(z)}{\bar{z}-\bar{w}} \rho_t (z) \rho_t(w)dA(z)dA(w) \right] dt
\\
&-\left[2 \int \left( \bp f m\pa Q -\pa \bp f  \right)\rho_t dA  -N\beta\int_{z\neq w} \frac{\bp f(z)}{z-w} \rho_t (z) \rho_t(w)dA(z)dA(w)\right]dt. 
\end{split} 
\end{align*} 
By taking the expectation on both sides of the above equation, letting $t\rightarrow \infty$ and using the definition \eqref{k-pt cor} we obtain 
\begin{align*}
0 =&-2 \int \left( \pa f m\bp Q -\pa \bp f  \right)R_{N,1}dA  +\beta \int \frac{\pa f(z)}{\bar{z}-\bar{w}} R_{N,2}(z,w)dA(z)dA(w) 
\\
&-2 \int \left( \bp f m\pa Q -\pa \bp f  \right)R_{N,1}dA  +\beta \int \frac{\bp f(z)}{z-w}  R_{N,2}(z,w)dA(z)dA(w).
\end{align*}
Note that here the condition $z\neq w$ is dropped due to the definition of $R_{N,2}$.
Since $f$ is a real-valued function, we have 
$$
\frac{\beta}{2} \int \frac{\bp f(z)}{z-w}  R_{N,2}(z,w)dA(z)dA(w)= \int \left( \bp f m\pa Q -\pa \bp f  \right)R_{N,1}dA.
$$
Moreover, since $\bp f$ is arbitrary, after integration by parts in the last term on the right-hand side, we conclude that  
\begin{equation*}
\frac{\beta}{2} \int \frac{R_{N,2}(z,w)  }{z-w} dA(w)= m\pa Q(z) R_{N,1}(z) + \pa R_{N,1}(z),
\end{equation*}
which completes the proof. $\Box$\\

\begin{example}
Recall that the \textit{elliptic Ginibre ensemble} is a one parameter family of 2D Coulomb gases where the external potential is given by 
	$$
	Q(z)=\frac{1}{1-\tau^2} \left(  |z|^2-\frac{\tau}{2}(z^2+\bar{z}^2 )  \right), \quad \tau \in (0,1).
	$$  	 
	It is well-known that the elliptic Ginibre ensemble interpolates between the Ginibre ensemble ($\tau=0$) and the GUE ($\tau \uparrow 1$) for $\beta=2$, see \cite{fyodorov1997almost}.	
	We remark that \eqref{diff Qz} gives the dynamical interpretation of this interpolation, valid for all $\beta>0$.
	More precisely, note that the system of diffusion processes for the elliptic Ginibre ensemble is given as
	\begin{align*}
	dx_j&=\sqrt{2} \, dB_j-m\frac{2x_j}{1+\tau} dt+\beta \sum_{k: k\not= j}^N \frac{x_j-x_k}{(x_j-x_k)^2+(y_j-y_k)^2}dt; 
	\\
	dy_j&=\sqrt{2} \, d\tilde{B}_j-m\frac{2y_j}{1-\tau}dt+\beta \sum_{k: k\not= j}^N \frac{y_j-y_k}{(x_j-x_k)^2+(y_j-y_k)^2}dt.  
	\end{align*}
	Observe that for every $j$, as $\tau \uparrow 1$, we have $y_j \sim \exp \left[ -\frac{2m}{1-\tau}t \right]$, which implies that the ensemble lies on the real line. Moreover, the diffusion process on the real line is given as 
	$$
	dx_j=\sqrt{2}dB_j-m\,x_jdt+\beta \sum_{k: k\not= j}^N \frac{1}{x_j-x_k}dt,
	$$
	which coincides with the one for $\beta$-GUE dynamics, see \eqref{lamba diff V} below.
\end{example}	

\subsection{Dynamics for 1D Coulomb gases}\label{Ito Sec1D}

In this subsection, we study the crossover regime for Coulomb gases confined on the real line. 
Note that our approach differs from \cite{MR3078021}.
For a given potential $V: \R \rightarrow \R$, we consider a system of particles labelled by $\bfs{\lambda}=(\lambda_1,\cdots,\lambda_{N}) \in \R^N$, having the joint probability density
\begin{align}\label{jpdf 1D} 
d\mathbf{P}_N(\bfs{\lambda})=  p_N(\bfs{\lambda})
\prod_{j=1}^{N}d\lambda_j\ , \quad
p_N(\bfs{\lambda})=\frac{1}{ Z_N}   \prod_{j>k=1}^N |\lambda_j-\lambda_k|^{\beta} e^{- m\sum_{j=1}^N V(\lambda_j)}\ ,
\end{align}
with normalising constant $Z_N$.
The corresponding $k$-point correlation functions are defined as in \eqref{k-pt cor}, with the corresponding counting function on the real line.

In analogy to \eqref{diff Qz}, let us consider the following dynamical system on the real line: 
\begin{equation} \label{lamba diff V}
d\lambda_j(t)=\sqrt{2} \, dB_j(t)-mV' \big( \lambda_j(t) \big) dt +\beta \sum_{k: k\not= j}^N \frac{1}{\lambda_j(t)-\lambda_k(t)}dt.
\end{equation}
The well-posedness of the above system with different assumptions on $V$ and $\beta$ has been studied by several authors, see e.g., \cite{cepa1997diffusing, MR3078021, rogers1993interacting, graczyk2014strong}.
For every time $t$, let  $p_N(t,\lambda)$ be the corresponding joint probability density function for the system \eqref{lamba diff V}. 
Then, as in Subsection~\ref{ITO Sec}, one can prove that the limiting stationary density function $p_N(\lambda):=\lim_{t\to\infty}p_N(t,\lambda)$ is given by 
\eqref{jpdf 1D}.

The following relation between the 1- and 2-point function corresponding to Theorem \ref{Ward} holds:
\begin{proposition}\label{Loop-1DProp}
Given the Gibbs measure \eqref{jpdf 1D}, with a smooth $C^2$ potential $V$ that satisfies $V(\lambda)>\log \lambda$ for $\lambda\to\infty$, the following relation holds for every finite-$N$:
\begin{equation} \label{Loop 1D}
\beta
\int_\R \frac{R_{N,2}(\lambda,\lambda')}{\lambda-\lambda'}  d\lambda'=   m\,V'(\lambda) R_{N,1}(\lambda)+R_{N,1}'(\lambda). 
\end{equation}
\end{proposition}

We note the difference in factor $1/2$ on the left-hand side compared to \eqref{LOOP} which is because we are in 1D now.
\begin{proof}
Let $f$ be a (real-valued) smooth function defined on the real line and set 
\begin{equation}
 \rho_t(\lambda):= \frac{1}{N} \sum_{j=1}^{N} \delta (\lambda-\lambda_{j}(t) )
\end{equation}
to be the normalised  time dependent one-point counting function on the real line. 
By It\^{o}'s formula and \eqref{lamba diff V} we obtain 
\begin{align*}
&d \int f(\lambda) \rho_t(\lambda)d\lambda- \frac{ \sqrt{2}}{N} \sum_{j=1}^{N} f'(\lambda_{j}) \, dB_j(t)
\\
=&\int_\R \Big( f''(\lambda)-mf'(\lambda) V'(\lambda) \Big) \rho_t(\lambda)d\lambda+N\beta \int_{\lambda \neq \lambda'} \frac{f'(\lambda)}{\lambda-\lambda'}\rho_t(\lambda)\rho_t(\lambda') d\lambda d\lambda'.
\end{align*}
Therefore, taking expectation values, using the definition according to \eqref{k-pt cor}, and letting $t \rightarrow \infty$, we obtain
\begin{align*}
\beta \int \frac{f'(\lambda)}{\lambda-\lambda'}R_{N,2}(\lambda,\lambda') d\lambda d\lambda'&= 
\int_\R \Big( f'(\lambda) V'(\lambda)-f''(\lambda) \Big)R_{N,1}(\lambda)d\lambda
\\
&= \int_\R  f'(\lambda) \Big( V'(\lambda) R_{N,1}(\lambda)+R_{N,1}'(\lambda) \Big) d\lambda,
\end{align*}
which leads to \eqref{Loop 1D}. $\Box$

\end{proof}


\section{Ward Identities in 2D and 1D} \label{WARD Sec}

\subsection{Ward identities in 2D} \label{Ward iden sub}
In this subsection, we discuss Ward identities for 2D Coulomb gases. They have been utilized already to derive the equation for the density function 
\eqref{Laplacian growth} for 2D Coulomb gases, with standard scaling \eqref{m scale1}, see \cite{ zabrodin2006large} and \cite[Chapter 39]{akemann2011oxford}. We adapt the proof presented in \cite{AKM14,AHM15} to derive the appropriate Ward identity for the 2D Coulomb gas distributed according to \eqref{2D Cou m}. 
For an alternative proof using so-called integration by parts see also \cite{ameur2018random,bauerschmidt2016two}, and for the general form of Ward identities we refer the reader to \cite[Appendix 6]{KM13}. The proof for the Ward identities in 1D presented in the next subsection  follows along the very same lines as in this subsection and we will not give much further details there.

For a test function $\psi \in C_0^\infty(\C)$ and $\bfs{\zeta}=(\zeta_1, \dots, \zeta_N)$, let us denote
\begin{align} \label{Ward functional123}
\begin{split}
\RN{1}_N[\psi](\bfs{\zeta})&=\frac{1}{4} \sum_{j,k:j\neq k}^N \frac{\psi(\zeta_{j})-\psi(\zeta_{k})}{\zeta_{j}-\zeta_{k}},
\\  
\RN{2}_N[\psi](\bfs{\zeta})&= m\sum_{j=1}^{N} \pa Q(\zeta_j) \psi(\zeta_{j}),
\\ 
\RN{3}_N[\psi](\bfs{\zeta})&= \sum_{j=1}^N \pa \psi(\zeta_{j}),
\end{split}
\end{align}
and define Ward's (stress energy) functional $W_N^+$ as 
\begin{equation}\label{Ward functional}
W_N^+[\psi]= \beta \, \RN{1}_N[\psi]- \RN{2}_N[\psi]+\RN{3}_N[\psi].
\end{equation}
From now on, we write $\E_N$ for the expectation with respect to \eqref{2D Cou m}. We first prove the following form of Ward's identity:
\begin{lemma}\label{EW=0LA}
For the definitions \eqref{Ward functional123} and \eqref{Ward functional} the following expectation value holds:
\begin{equation}\label{EW=0}
\E_N W_N^+[\psi]=0.
\end{equation} 
\end{lemma}
\begin{proof}
	By definition, the partition function $Z_N$ is given as
	\begin{equation*}
	Z_N= \int_{\mathbb{C}^N} \exp\left[ \beta \sum_{j > k=1}^N \log |\eta_{j}-\eta_{k}| - m\sum_{j=1}^{N} Q(\eta_{j}) \right] \prod_{j=1}^{N}dA(\eta_j).
	\end{equation*}
	For a fixed sequence $\bfs{\zeta}=(\zeta_1,\cdots,\zeta_N)$ and a positive constant $\ve$, let us denote
	\begin{equation}\label{eta zeta}
		\bfs{\eta} := (\eta_1, \cdots, \eta_N), \quad \eta_{j}:=\zeta_j+\frac{\ve}{2}\psi(\zeta_{j}).
	\end{equation}
	Then as $\ve \rightarrow 0$, we have
	\begin{align*}
	\log|\eta_{j}-\eta_{k}|= \log|\zeta_{j}-\zeta_{k}|+\frac{\ve}{2}\Re \, \frac{\psi(\zeta_{j})-\psi(\zeta_{k})}{\zeta_{j}-\zeta_{k}}+ O(\ve^2),
	\end{align*} 
	thus leading to 
	\begin{equation*}
	\beta \sum_{j > k=1}^N \log |\eta_{j}-\eta_{k}|= \beta \sum_{j >k=1}^N \log |\zeta_{j}-\zeta_{k}|+\ve \beta \, \Re \, \Big[ \RN{1}_N[\psi](\bfs{\zeta}) \Big]+O(\ve^2). 
	\end{equation*}
	Since we assume that $Q$ is $C^2$-smooth, we have
	\begin{equation*} 
	m\sum_{j=1}^{N} Q(\eta_{j})=m\sum_{j=1}^{N} Q(\zeta_{j})+\ve \Re \, \Big[ \RN{2}_N[\psi](\bfs{\zeta}) \Big]+O(\ve^2).
	\end{equation*}
	Note that due to the fact that the Jacobian of \eqref{eta zeta} is given as
	\begin{align*}
	dA(\eta_{j})&=\left( \left| 1+\frac{\ve}{2}\pa \psi(\zeta_{j}) \right|^2- \left|\frac{\ve}{2}\bp \psi(\zeta_{j}) \right|^2  \right) dA(\zeta_{j})
	\\
	&= \Big( 1+\ve \Re \, \pa \psi(\zeta_{j})+O(\ve^2) \Big) dA(\zeta_{j}), 
	\end{align*}
	we have 
	\begin{equation*}
	\prod_{j=1}^{N}dA(\eta_j)= \left(1+\ve\Re \, \Big[ \RN{3}_N[\psi](\bfs{\zeta}) \Big]+O(\ve^2) \right)  \prod_{j=1}^{N} dA(\zeta_j). 
	\end{equation*}
	Combining all the above equations, we obtain
	\begin{align*}
	Z_N
	=\int_{\mathbb{C}^n}
	\prod_{j>k=1}^N |\zeta_j-\zeta_k|^{\beta} e^{-m\sum_{j=1}^N Q(\zeta_j)} 
	\left(1+\ve\Re\left[ W_N^+[\psi](\bfs{\zeta}) \right]+O(\ve^2) \right) \prod_{j=1}^{N} dA(\zeta_{j}).
	\end{align*}
	Observe that since the partition function $Z_N$ does not depend on $\ve$, the coefficient of $\ve$ in the right-hand side of above identity is zero, i.e.,
	$\Re\, \E_N W_N^+[\psi]=0.$ 
	Now \eqref{EW=0} follows by same argument with $i \psi$. $\Box$
\end{proof}

Next we prove Theorem \ref{Ward} equation \eqref{LOOP}  using the previous Lemma \ref{EW=0LA} equation \eqref{EW=0}. 
\begin{proof}
Suppose that the point $\zeta \in \C$.  
Recall that $R_{N,k}$ is the $k$-point correlation function \eqref{k-pt cor} for the 2D Coulomb gas given by \eqref{2D Cou m}.
Let $\psi\in C_0^\infty(\C)$ be an arbitrary test function. By definition, we have
\begin{align*}
\E_N \RN{1}_N[\psi]&=\frac{1}{4}\int_{\C^2}\frac{\psi(\zeta)-\psi(\eta)}{\zeta-\eta} R_{N,2}(\zeta,\eta) dA(\zeta)dA(\eta)
\\
&=\frac{1}{2}  \int_{\C}\psi(\zeta)\int_{\C} \frac{R_{N,2}(\zeta,\eta) }{\zeta-\eta} dA(\eta) dA(\zeta) ,
\\
\E_N \RN{2}_N[\psi]&= m \int_{\C} \psi(\zeta) \pa Q(\zeta) R_{N,1}(\zeta) dA(\zeta),
\\
\E_N \RN{3}_N[\psi]&= \int_{\C} \pa \psi(\zeta) R_{N,1}(\zeta)dA(\zeta)= -\int_{\C} \psi(\zeta) \pa R_{N,1}(\zeta) dA(\zeta).
\end{align*}
Therefore \eqref{EW=0} is rewritten as
\begin{align*}
&\frac{\beta}{2}  \int_{\C}\psi(\zeta) \left( \int_{\C} \frac{R_{N,2}(\zeta,\eta) }{\zeta-\eta} dA(\eta) \right) dA(\zeta)
\\
=&  \int_{\C} \psi(\zeta) \Big( m\pa Q(\zeta) R_{N,1}(\zeta)+\pa R_{N,1}(\zeta) \Big) dA(\zeta).
\end{align*}
Since $\psi$ is an arbitrary test function, \eqref{LOOP} follows. 
\end{proof}

\subsection{Ward identities in 1D}\label{Ward1Dproof} 

For a test function $\psi \in C_0^\infty(\R)$ and $\bfs{\lambda}=(\lambda_1, \dots, \lambda_N)$, define the corresponding Ward stress energy functional 
\begin{equation}\label{Ward functional R}
\wh{W}_N^+[\psi]= 2\beta \, \RN{1}_N[\psi]- \RN{2}_N[\psi]+\RN{3}_N[\psi].
\end{equation} 
Here $\RN{1}_N, \RN{2}_N, \RN{3}_N$ are given as \eqref{Ward functional123} except that $\bfs{\zeta}$ and $Q$ are replaced with $\bfs{\lambda}$ and $V$. 
Following the proof in Subsection~\ref{Ward iden sub}, it is straightforward to obtain the same statement as in Lemma \ref{EW=0LA}:
\begin{equation}
\E_N \wh{W}_N^+[\psi]=0.
\end{equation}
By definition, we have that
\begin{align*}
\E_N \RN{1}_N[\psi]&=\frac{1}{4}\int_{\R^2}\frac{\psi(\lambda)-\psi(\lambda')}{\lambda-\lambda'} R_{N,2}(\lambda,\lambda') d\lambda d\lambda' ;
\\
\E_N \RN{2}_N[\psi]&= m\int_{\R} \psi(\lambda) V'(\lambda) R_{N,1}(\lambda) d\lambda;
\\
\E_N \RN{3}_N[\psi]&= -\int_{\R} \psi(\lambda) R'_{N,1}(\lambda) d\lambda.
\end{align*}
Therefore equation \eqref{Loop 1D} in Proposition \ref{Loop-1DProp} follows. 
\\

\section{The Saddle Point Method in 2D and 1D} \label{Saddlepoint}

In this section our approach will be more heuristic. First, we will calculate the free energy functional $F_N[\rho_N]$ for large but finite $N$, both for the 1D and 2D case together. In this way it will become clear, how the respective dimension $d=1,2$ enters. Furthermore, we will see how imposing the different scaling limits \eqref{beta n 2c} and 
\eqref{m scale1} leads to different limiting free energies \eqref{free} and \eqref{Frost Hamil}, respectively, that arise from a different order in $N$. 
Only after imposing the saddle point condition upon the limiting free energy functionals, we have to specify the dimension $d$. 
In 2D ($d=2$) we can use the Laplace operator to directly obtain an equation for the limiting density. In contrast, in 1D ($d=1$) we have to first pass over to the resolvent $G(z)$ or Stieltjes transform of the limiting density, 
to find a closed form equation that determines it, and then finally obtain the limiting density by taking the discontinuity along its support. 
We refer the reader to [Chapter 4,5]\cite{livan2017introduction} for the general concepts of the saddle point method in 1D and to \cite{zabrodin2006large} in 2D. 

\subsection{The free energy in 2D and 1D}

We begin by writing down the partition function for the Gibbs measures \eqref{2D Cou m} and \eqref{jpdf 1D} in a unified way,
\begin{equation}
\label{jpdf-gen}
Z_N=\int \prod_{j>k=1}^N |\zeta_j-\zeta_k|^{\beta} e^{-m\sum_{j=1}^N Q(\zeta_j)} \prod_{j=1}^{N}d\mu_d(\zeta_j).
\end{equation}
Here, for $d=2$ we integrate over $\mathbb{C}^N$ and $d\mu_2=dA$ is the area measure, that is the 2D Lebesgue measure over $\pi$, whereas for $d=1$ we integrate over $\mathbb{R}^N$ and $d\mu_1(\zeta)=d\zeta$ is the flat Lebesgue measure in 1D. Clearly, the integrand can be written as the exponential of the following 
energy function $E_N[ \bfs{\zeta}  ]$
\begin{equation}\label{EnergyN}
E_N[  \bfs{\zeta} ]:= m \sum_{j=1}^{N} Q(\zeta_{j}) - \beta \sum_{j>k=1}^N \log |\zeta_{j}-\zeta_{k}| \ .
\end{equation}
Our first goal is to change variables from the particle positions $\zeta_{j=1,\ldots,N}$ to the normalised one-point counting function $\rho_N(z)$ 
from \eqref{norm-count}
\begin{equation}\label{count}
\rho_N(z)=\frac{1}{N}\sum_{j=1}^N \delta^{(d)}(z-\zeta_{j})\ ,
\end{equation}
such that we can write
\begin{align} \label{Multiple int}
\begin{split}
Z_N &=\int \exp[-E_N[ \bfs{\zeta} ]]  \prod_{j=1}^{N}d\mu_d(\zeta_j) = 
\int \exp[-E_N[ \rho_N]]  J_N[\rho_N] \mathcal{D}[\rho_N]\\
&= \int \exp[-F_N[ \rho_N]] \mathcal{D}[\rho_N].
\end{split}
\end{align}
Here,  $F_N[\rho_N]$ is the free energy functional for large but finite $N$ we seek for, $\mathcal{D}[\rho_N]$ is the integration over the counting measure, and  $J_N[\rho_N]$ is the Jacobian that formally reads
\begin{equation} 
\label{JacobiN}
J_N[\rho_N]=\int \delta\left(\rho_N(z)-\frac{1}{N}\sum_{j=1}^N \delta^{(d)}(z-\zeta_{j})\right) \prod_{k=1}^N d\mu_d(\zeta_k)\ .
\end{equation}
It will be computed below for $N\gg1$, and its contribution to the free energy is called {\it entropy}. 
By standard thermodynamic arguments the ensemble will converge towards to the limiting density (equilibrium measure) that minimises the free energy (or maximises its, should it be negative). 

We begin by expressing the energy \eqref{EnergyN} in terms of the counting function \eqref{count}. For the first term we simply have 
\begin{equation*}
m\sum_{j=1}^{N} Q(\zeta_{j})=  N \int \rho_N(z)Q(z)d\mu_d(z)\ .
\end{equation*}
For the second term in \eqref{EnergyN} we can write, after symmetrising,
\begin{align*}
-\frac{\beta}{2}\sum_{j,k: j\neq k}^N \log |\zeta_{j}-\zeta_{k}|
=&-N^2\frac{\beta}{2} \int \rho_N(z)\rho_N(z')\log|z-z'|d\mu_d(z)d\mu_d(z')
\\
&+N\frac{\beta}{2} \int \rho_N(z) \log \ell(z) d\mu_d(z).
\end{align*}
Because the sum does not contain points at equal argument we have to subtract the diagonal contribution which is divergent. As we are only interested in the density on a global, macroscopic scale which is much larger than the mean particle distance, we have introduced a \textit{short-distance cut-off} $\ell(z)$ which may be position-dependent. This term is also called {\it self-energy}, and because the mean particle distance depends on the dimension $d$, in the bulk of the spectrum we have for large $N$
\begin{equation}
\ell(z) \simeq \left(\frac{1}{N\rho_N(z)}\right)^{1/d}\ ,\quad \textrm{with}\ d=1,2\ ,
\end{equation}
see e.g., \cite{livan2017introduction} for $d=1$ and \cite[Section 2]{zabrodin2006large} for $d=2$. Clearly this argument is not rigorous. The last ingredient we miss is the Jacobian \eqref{JacobiN} to be derived later, which for large but finite-$N$ reads
\begin{equation}\label{Jacobi}
J_N[\rho_N]= \exp \left[-N \int \rho_N(z) \log \rho_N(z) d\mu_d(z)-N \log N+\gamma_d N+o(N) \right].
\end{equation}
Here, $\gamma_d$ is some constant, see \cite[eq. (2.15)]{zabrodin2006large} for $d=2$, which is apparently unknown for $d=1$ \cite{livan2017introduction}.
Collecting all contributions we obtain the following result for the free energy functional at large-$N$: 
\begin{align}
\begin{split}
F_N[\rho_N]\approx&\ m N \int Q(z) \rho_N(z) d\mu_d(z)
\\
&-\frac{\beta}{2}N^2 \int \rho_N(z)\rho_N(z')\log|z-z'|d\mu_d(z)d\mu_d(z')
\\
&+\frac{N}{2d}(2d-\beta) \int \rho_N(z) \log \rho_N(z) d\mu_d(z)
\\
&+C_d \left(\int \rho_N(z) d\mu_d(z) -1\right).
\label{FN}
\end{split}
\end{align}
We have added a term that ensures the correct normalisation of the density, and the constant $C_d$ is called \textit{Lagrange multiplier}. 
For simplicity we have suppressed all other constants and $o(N)$ terms here, as they will not play any r\^ole later. 

Notice that for $\beta=2$ in $d=1$ and for $\beta=4$ in $d=2$ the term in the third line of \eqref{FN} is absent, cf. \cite{zabrodin2006large,itoi1997universal}, respectively. This leads to the well known fact that for these particular values of $\beta$ the free energy can be expanded in powers of $1/N^2$ also called genus expansion, whereas the expansion is in powers of $1/N$ in all other cases, see \cite{BorotGuionnet} for a recent work. 

Before we turn to the different large-$N$ limits let us briefly derive the entropic factor $J_N[\rho_N]$ which can be computed by simple combinatorial arguments, see e.g.,\cite{allez2012invariant2,zabrodin2006large}. By definition, $J_N[\rho_N]$ is the number of microstates which are compatible with a given local density function 
$\rho_N(z)$.  First, note that we may assume that almost all the particles are confined inside a large square for $d=2$ (line for $d=1$) since $Q$ is sufficiently large near infinity. We divide this square (line) into $N$ equal cells $r_{j=1,\cdots, N}$ and set 
$ N_j := N \int_{r_j} \rho_N(z) d\mu_d(z),$ which implies that $N_j/N$ is the local density in the cell $r_j$. Note that $N=N_1+\cdots+N_N$ and by definition, $J_N[\rho_N]$ is asymptotically the number of cases that each cell $r_j$ is occupied by $N_j$. Then by Stirling's formula, for large-$N$ we have 
$$
\frac{N!}{N_1 ! \dots N_N !} \sim \left( \frac{N_1}{N} \right)^{-N_1} \cdots \left( \frac{N_N}{N} \right)^{-N_N}  (2\pi)^{-N/2} \frac{\sqrt{N}}{\sqrt{ N_1} \cdots \sqrt{N_N}}.
$$
Taking the logarithm 
we obtain 
\begin{align*}
\log \frac{N!}{N_1 ! \dots N_N !} &\sim -\sum_{j=1}^{N}N_j \log \frac{N_j}{N} -\frac{N}{2} \log 2\pi+\frac{1}{2} \Big(    \log N- \sum_{j=1}^{N} \log N_j        \Big)
\\
&=-N \sum_{j=1}^{N} \frac{N_j}{N} \log \frac{N_j}{N} -\frac{N}{2} \log 2\pi-\frac{1}{2} \Big((N-1)  \log N+\sum_{j=1}^{N} \log \frac{N_j}{N}        \Big).
\end{align*}
Therefore, in the large-$N$ limit  we obtain \eqref{Jacobi}. Notice that the last term on the right-hand side above will also contain the density, but is of sub-leading order. 
We also refer the reader to \cite{livan2017introduction,dean2008extreme} for a different approach 
using the integral representation for the delta function in \eqref{JacobiN}.

Let us discuss the two different scaling large-$N$ limits  \eqref{beta n 2c} and 
\eqref{m scale1} of the free energy \eqref{FN}, starting with the more standard limit \eqref{m scale1}.
\begin{enumerate}[(i)]
	\item First, let $m=\beta N/2$ and $\beta=O(1)$ be fixed according to \eqref{m scale1} which is the standard scaling limit for $\beta$-ensembles. Then the leading contribution of the free energy \eqref{FN} is of order $N^2$ (from the first and second line) and results from the contribution of the energy terms only. Assuming that the Lagrange multiplier $C_d$ is of order unity we obtain
	\begin{align*}
F[\rho]&=	
	\lim_{N\to\infty}\frac{2F_N[\rho_N]}{\beta N^2}\\
	&=\int Q(\zeta) \rho(\zeta) d\mu_d(\zeta)-\int \log {|\zeta-\eta|} \rho(\zeta) \rho(\eta) d\mu_d(\zeta) d\mu_d(\eta)\ .
	\end{align*}
	It agrees with the functional \eqref{Frost Hamil}.
	The equation determining the density $\rho_*$, that minimises the free energy in either limit, is given by the \textit{saddle point equation}, a necessary condition to have an extremum. We therefore require the functional derivative of $F$ to vanish at the equilibrium density $\rho_*$:
\begin{equation}\label{SP-F-standard}
0=\frac{\delta F(\rho)}{\delta \rho(\zeta)}\Big|_{\rho=\rho_*}=Q(\zeta)-2 \int\log |\zeta-\eta| \, \rho_*(\eta) \, d\mu_d(\eta)\ .
	\end{equation} 
From a heuristic point of view we can easily see that this indeed minimises the free energy. Taking a second functional derivative that we regularise by choosing $\xi$ slightly away from $\zeta$, we have 	
	\begin{equation}\label{SP-max}
\frac{\delta^2 F(\rho)}{\delta \rho(\zeta)\delta \rho(\xi)}\Big|_{\rho=\rho_*}=-2 \log |\zeta-\xi| >0 \ ,
	\end{equation} 
which is clearly positive, as for $\xi\approx\zeta$ the logarithm becomes negative. The fact that the solution of \eqref{SP-F-standard} is a minimum can be made rigorous and we refer to \cite{johansson1998fluctuations} and \cite{HM13} for references for $d=1,2$, respectively.
	
	\item Second, let $m=1$ and $\beta= 2c /N$ for some $c \in (-d,\infty)$  which agrees with our proposed scaling \eqref{beta n 2c} for $d=2$ and \cite{allez2012invariant} for $d=1$. 
	In this case we have that both energy and entropy contribute and the leading order in \eqref{FN} is now rather $N$. Therefore we obtain instead 
	\begin{align*}
F_c[\rho]=&\lim_{N\to\infty}	\frac{F_N[\rho_N]}{N}\\
=&  \int  Q(\zeta) \rho(\zeta) d\mu_d(\zeta)-c \int\log {|\zeta-\eta|} \rho(\zeta) \rho(\eta) d\mu_d(\zeta) d\mu_d(\eta)
	\\
	&+\int\rho(\zeta) \log \rho(\zeta) d\mu_d(\zeta)\ ,
	\end{align*}
	which agrees with the free energy claimed in \eqref{free}.
Here, the corresponding {saddle point equation} reads 
\begin{align}
\begin{split}
\label{SP-F}
0&=\frac{\delta F_c(\rho)}{\delta \rho}\Big|_{\rho=\rho_c}
\\
&=Q(\zeta)-2c \int \log |\zeta-\eta| \, \rho_c(\eta) \, d\mu_d(\eta) + \log \rho_c(\zeta)+1.
\end{split}
\end{align}
	The second functional derivative that decides whether we have a minimum or a maximum leads to 
	\begin{equation}\label{SP-max-c}
\frac{\delta^2 F_c(\rho)}{\delta \rho(\zeta)\delta \rho(\xi)}\Big|_{\rho=\rho_c}=-2c \log |\zeta-\xi| +\frac{\delta^{(d)}(\zeta-\xi)}{\rho_c(\zeta)} \ .
	\end{equation} 	
	Due to our regularisation $\xi\approx\zeta$ the logarithm becomes negative and, ignoring the second term at this scale, we obtain a minimum for $c>0$ and a maximum for $c<0$. 
This statement has been made rigorous for the Gaussian plus linear potential for $d=2$ in \cite{caglioti1992special}. 	
	For $c=0$ we do not have an extremum, and the solution of \eqref{SP-F} for $c=0$ leads to $\rho_{c=0}(\zeta)\sim\exp[-Q(\zeta)]$, as is expected for non-interacting particles. 
	\end{enumerate}	

\subsection{Saddle point equation for the density in 2D}
 If we want to transform the saddle point equation into a closed differential equation for $\rho_c$ we have have to distinguish now the cases $d=1$ and $d=2$. While $d=1$ is considerably more complicated, passing through the resolvent as explained in the next subsection, $d=2$ is in principle very simple. This is due to the fact that the Laplacian acting on the logarithm gives a Dirac delta, which in our convention \eqref{rad Lap} reads $\Delta\log|z|=\frac{\pi}{2}\delta^{(2)}(z)$. In the limit (i) above we thus obtain from \eqref{SP-F-standard} that
 \begin{equation}
 \label{MFE-1}
 \Delta Q(\zeta)-\rho_*(\zeta)=0\ ,
 \end{equation}
which has to hold on the limiting support, the droplet, as claimed in \eqref{Laplacian growth}.
For the limit (ii) from \eqref{SP-F} we get
\begin{equation}
\label{Laplace-1}
\Delta Q(\zeta) -c \, \rho_c(\zeta) +  \Delta\log \rho_c(\zeta) =0\ ,
\end{equation}
which is supported apriori on the entire complex plane. This is the  mean field equation \eqref{MFE}. As already explained in the introduction we have been unable to solve this equation analytically. We refer again to the numerical solution for two examples presented there for radially symmetric potentials, to which we turn now.

For simplicity, we focus on the potentials $Q(z)=|z|^{2\alpha}/2$, where $\alpha\geq1$. Recall that the radial part $g_c(|z|):=\rho_c(z)/\pi$ of the crossover density satisfies 
\begin{equation} 
\label{radial-Freud1}
4\pi c \, r \, g_c(r)=2\alpha^2 r^{2\alpha-1} +\left(r \, (\log g_c(r))'\right)'.
\end{equation}
Based on this it is not difficult to see the asymptotic behaviour in $c$ for $c\to\infty$ and $c\to0$ as quoted in \eqref{c-limits}. Namely, for $c\to\infty$ in order to get a finite answer on left- and right-hand side we need that $g_c(r)\sim1/c$. Neglecting the last term in \eqref{radial-Freud1}, which self-consistently leads from \eqref{Laplace-1} to \eqref{MFE-1}, we are lead to 
\begin{equation*}
g_c(r)\sim \frac{\alpha^2}{2\pi c} r^{2\alpha-2} \ .
\end{equation*}
The limiting support on a disc of radius $b$ simply follows by imposing the normalisation condition
\begin{equation*}
\frac{1}{2\pi}=\int_0^b dr r g_c(r)= \frac{\alpha}{4\pi c} b^{2\alpha} \ ,
\end{equation*}
which leads to $b=(2c/\alpha)^{1/(2\alpha)}$ as claimed in \eqref{c-limits}. For $c\to0$ we then obtain 
$g_c(r)\sim k\exp[-r^{2\alpha}/2]$ and we simply have to compute the normalisation constant $k$ from 
\begin{equation*}
\frac{1}{2\pi}=\int_0^\infty dr \, r \, k \, \exp(-r^{2\alpha}/2)=2^{\frac{1}{\alpha}-1}\Gamma(1+1/\alpha)\ .
\end{equation*}
This implies $k=1/(\pi 2^{1/\alpha}\Gamma(1+1/\alpha))$ as claimed in \eqref{c-limits}. Of course, the statement $g_c(r)\sim k\exp[-Q(r)]$ holds for more general radially symmetric potentials in the limit $c\to0$, with the difficulty to determine the normalisation constant $k$ for a given $Q$.

When stating our main results we have derived already the asymptotic behaviour \eqref{large-r} for radially symmetric potentials for large radii $r\to\infty$. Let us add a few remarks here about a possible expansion for small $r$. 
Assuming that $g_c(0)\neq0$, which will be true for $c$ not too large, let us denote 
\begin{align}
\begin{split}
g_c(r):&=g_c(0)\,\exp\left( \sum_{j=1}^{\infty} a_j\,r^{2j} \right)
\end{split}
\end{align}
in order to obtain an expansion for small $r$.
By inserting the above expression in \eqref{radial-Freud1} and comparing the coefficients, one can iteratively express the $a_j$ through $g_c(0)$, thus leading to
\begin{align}
\begin{split}
a_1&=
\begin{cases}
\pi cg_c(0)-\frac12 &\text{if }\alpha=1
\\
\pi c g_c(0) &\text{if }\alpha=2
\end{cases}, \\
a_2&=
\begin{cases}
\frac14(\pi c g_c(0))^2-\frac18\pi c g_c(0) &\text{if }\alpha=1
\\
\frac14(\pi c g_c(0))^2-\frac12 &\text{if }\alpha=2 
\end{cases}.
\end{split}
\end{align}
We can immediately compare this to what we have obtained in the previous paragraph for $c\to0$, where we found 
$g_{c=0}(0)=1/(\pi 2^{1/\alpha}\Gamma(1+1/\alpha))$. Therefore $cg_{c=0}(0)$ is vanishing in the limit $c\to0$ and we have
\begin{align} \label{Freud asmp r0}
\begin{split}
g_c(r)=\begin{cases}
g_c(0) \left(1-\frac12 r^2+O(r^6) \right)  &\text{if }\alpha=1
\vspace{0.5em}
\\
g_c(0) \left(1-\frac12 r^4+O(r^6) \right) &\text{if }\alpha=2
\end{cases}\ ,
\end{split}
\end{align}
which agrees with $g_{c=0}(r)=g_{c=0}(0)\exp(-r^{2\alpha}/2)$. In fact it is not difficult to see that for $c\to0$ all other coefficients vanish and $a_j=-\frac12 \delta_{j,\alpha}$. For $c\to\infty$ the assumption $g_c(0)\neq0$ breaks down unless $\alpha=1$, where we obtain $g_c(0)\sim \frac{1}{2\pi c}$. 

This ends our short survey on the asymptotic behaviour of the solution of \eqref{radial-Freud1} in $c$ and radius $r$ for $\alpha=1,2$.

\subsection{Saddle point equation for the density in 1D}

We will now discuss the saddle point equation in 1D where we will focus on the second limit (ii) above, using \eqref{SP-F}.
It turns out that in order to determine the solution for the density it is more convenient to pass through the resolvent to be defined in \eqref{resolvent} below, and we will illustrate this through an example. 

Denoting the particle positions by $\lambda\in \R$ (instead of $\zeta$), the real potential $Q$ by $V$, and writing $\rho_c$ for the limiting density function on $\R$ we may differentiate \eqref{SP-F} with respect to $\lambda$: 
\begin{equation} \label{MFE 1D}
0=V'(\lambda)-2c\, \textrm{Pr} \int_{\R} \frac{\rho_c(\lambda')}{\lambda-\lambda'}d\lambda' +\Big( \log \rho_c (\lambda) \Big)'.
\end{equation} 
Here, we have to take the principal value (Pr) of the real integral.
Let us  denote by $G_c(z)$ the Stieltjes transformation (or resolvent) of the limiting density $\rho_c$. It is given as 
\begin{equation}
\label{resolvent}
G_c(z):=\int_{\R} \frac{\rho_c(\lambda)}{\lambda-z}d\lambda, \quad z \in \C \sm \R. 
\end{equation}
From the normalisation of the density we can see that at large argument it behaves as $G_c(z)\sim-\frac{1}{z}$.
Our next goal is to derive a closed form equation for the resolvent. 
To that aim we multiply equation \eqref{MFE 1D} by $\rho_c(\lambda)/(\lambda-z)$ and integrate over the real line, to obtain 
\begin{equation}
\label{SP-int}
0=\int_{\R} \frac{V'(\lambda)\rho_c(\lambda)}{\lambda-z}d\lambda 
-2c\, \textrm{Pr} \int_{\R^2} \frac{\rho_c(\lambda)}{\lambda-z} \frac{ \rho_c(\lambda')}{\lambda-\lambda'}d\lambda'  d\lambda 
+\int_{\R} \frac{\rho'_c(\lambda)}{\lambda-z}d\lambda .
\end{equation}
The last term can be most easily rewritten, after using integration by parts:
\begin{equation}
\label{Gprime}
\int_{\R}\frac{\rho'_c(\lambda)}{\lambda-z} d\lambda ={G_c}'(z).
\end{equation}
For the second term with the double integral we use the identity 
$$
\frac{1}{\lambda-z} \frac{1}{\lambda-\lambda'}=-\left(\frac{1}{\lambda-z}-\frac{1}{\lambda-\lambda'}\right)\frac{1}{\lambda'-z},
$$
to observe that 
$$
\textrm{Pr} \int_{\R^2} \frac{\rho_c(\lambda)}{\lambda-z} \frac{ \rho_c(\lambda')}{\lambda-\lambda'}d\lambda'  d\lambda = 
-\left(\int_{\R} \frac{\rho_c(\lambda)}{\lambda-z} d\lambda \right)^2
+\textrm{Pr} \int_{\R^2} \frac{\rho_c(\lambda)}{\lambda-\lambda'} \frac{ \rho_c(\lambda')}{\lambda'-z}d\lambda'  d\lambda,
$$
after dropping the principal value in the first term on the right-hand side. Observing that both integrals agree after a change of variables, we thus obtain 
\begin{equation}
\label{Gsquare}
\textrm{Pr} \int_{\R^2} \frac{\rho_c(\lambda)}{\lambda-z} \frac{ \rho_c(\lambda')}{\lambda-\lambda'}d\lambda'  d\lambda = -\frac12
G_c^2(z).
\end{equation}

For the remaining first term in \eqref{SP-int} we have to make a further approximation. In the standard scaling limit \eqref{m scale1} for particles on $\R$ the limiting density has a compact support, for sufficiently confining potentials. Then, one can easily define a contour in the complex plane that encircles the support, but not the point $z\in \C \sm \R$. 
In our case only for large $c$ we know that the limiting support localises on the semi-circle or its generalisation. But also for small $c$ the density typically decreases exponentially for large arguments, e.g. for the class of Freud potentials. We therefore assume that we may truncate the integral on a large interval $J$, making an exponentially small error. At the end of the calculation we can then take the limit $J\to\R$. 
Note that we can allow $V'$ to have poles on the real line, as in one of our examples below, but not to have a cut that extends to infinity.

Let us therefore define a contour $\mathcal{C}_J$ that encircles $J$ in counter-clockwise fashion and does not contain the point $z\in \C \sm \R$. We can then use the residue theorem to arrive at
\begin{align}
\label{VprimeG}
\begin{split}
\int_{\R}  \frac{V'(\lambda)\rho_c(\lambda)}{\lambda-z}d\lambda&\approx\int_{J}\rho_c(\lambda) \oint_{\mathcal{C}_J}\frac{1}{w-\lambda} \frac{V'(w)}{w-z}  \frac{dw}{2\pi i} d\lambda 
\\
&=\oint_{\mathcal{C}_J} G_c(w) \frac{V'(w)}{w-z}\frac{dw}{2\pi i} 
\\
&=\oint_{\infty} \frac{G_c(w) V'(w)}{w-z}\frac{dw}{2\pi i} + G_c(z) V'(z)
\end{split}
\end{align}
In the second step we have interchanged integrations, to obtain an expression depending only on the resolvent. In the last step we have pulled the contour to infinity, picking up the contribution from the pole at $z$. Here nothing depends any more on the regularising integral $J$. 
Combining all the above equations \eqref{Gprime}, \eqref{Gsquare} and \eqref{VprimeG} we obtain the following closed form equation, assuming that our prescribed cut-off procedure can be made rigorous:
\begin{equation} \label{Rel Green}
0=\oint_{\infty}  \frac{G_c(w) V'(w)}{w-z}\frac{dw}{2\pi i}
+ G_c(z) V'(z)
+ c \, G_c^2(z)+{G_c}'(z).
\end{equation}
The remaining contribution at infinity is not easily evaluated for a general potential $V$.
In the standard limit \eqref{m scale1} at fixed $\beta$ an ansatz can be made for the compact support to consist of a finite union of intervals. 

\begin{example}Here, we consider the Gaussian potential $V_G(w)=w^2/2$, cf. \cite{allez2012invariant}. In that case we may exploit the behaviour of the resolvent \eqref{resolvent} at infinity, $G_c(w)\sim-1/w$, to evaluate the contour integral at infinity to give unity. We thus have 
\begin{equation} \label{Rel Green Gauss}
0=1
+ G_c(z) z
+ c \, G_c^2(z)+{G_c}'(z)\ .
\end{equation}
which agrees with the equation found in \cite{allez2012invariant}. There, the derivation of \eqref{Rel Green Gauss} could be made rigorous using \cite{rogers1993interacting}. In \cite{allez2012invariant} this equation was solved for the density $\rho_c$ by taking the discontinuity of $G_c$ along the real line, see \eqref{Inversion} below. We will illustrate this procedure with a more general example below. Consequently, the authors found the following explicit formula for the crossover density $\rho_c$ in terms of the parabolic cylinder function $D$:
\begin{equation} \label{Gauss-Wigner}
\rho_c(\lambda)=\frac{1}{\sqrt{2\pi} \, \Gamma(1+c) } |D_{-c} (i\lambda) |^{-2} .
\end{equation}
\end{example}

\begin{example} 
We now slightly extend the previous example by considering a Gaussian potential with an additional logarithmic singularity. A similar case was considered on $\R_+$ in \cite{allez2012invariant2}.
For any real parameter $a>-1$, let us consider the potential  
\begin{equation}\label{Gauss+a}
V_a(x)= \frac12 x^2-a\log|x|\ .
\end{equation}
The contour integral in \eqref{Rel Green} at infinity can be solved as in the previous example, as the pole of $V_a'$ does not contribute there. 
Therefore, the resolvent $G_{c,a}$ satisfies the following Riccati type equation 
\begin{equation} \label{Rect G}
0=1
+ G_{c,a}(z) \left(z-\frac{a}{z}\right)
+ c \, G_{c,a}^2(z)+{G_{c,a}}'(z)\ .
\end{equation}
Let us denote  
\begin{equation}
G_{c,a}(z)=:\frac{1}{c}\Big( \log u(z) \Big)'.
\end{equation} 
Then the ODE \eqref{Rect G} can be rewritten in terms of the new function $u(z)$ as  
\begin{equation}
\label{udiff}
0=cu(z)+(z-a/z)u'(z)+u''(z)\ .
\end{equation}
A further change of variables to $w=-z^2/2$, with 
$u(z)=:f(-z^2/2=w)$,  casts this into the form of  Kummer's differential equation 
\begin{equation} \label{Kummer}
w \, f''(w)+\left( \frac{1-a}{2}-w \right)f'(w)-\frac{c}{2}f(w)=0.
\end{equation}
Moreover, since $G_{c,a}(z)\sim -1/z$ near infinity, we have $u(z)\sim |z|^{c}$ and thus 
\begin{equation} \label{Kummer inf}
f(w) \sim |w|^{-c/2}, \quad |w|\rightarrow \infty\ .
\end{equation}
It is well-known that the solution of \eqref{Kummer} satisfying \eqref{Kummer inf} is uniquely determined (up to a multiplicative constant) and reads  
\begin{equation} \label{Kummer ftn}
f(w)= U\left( \frac{c}{2}, \frac{1-a}{2}, w  \right),
\end{equation}
see e.g., \cite[p.322]{olver2010nist}. Here, $U$ is  Kummer's (confluent) hypergeometric function given as the analytic continuation of the integral representation 
\begin{equation}
U(\alpha,\gamma,z)=\frac{1}{\Gamma(\alpha)} \int_{0}^{\infty} e^{-z t} t^{\alpha-1} (1+t)^{\gamma-\alpha-1} dt, \quad (\Re\,\alpha>0). 
\end{equation} 
Now let us introduce  
\begin{align}
\begin{split}
y(z):=e^{z^2/4} z^{-a/2}  u(z),
\end{split}
\end{align}
which, upon using \eqref{udiff}, leads to 
\begin{equation} \label{ODE y}
y''(z)+\Big[ c-\frac{1}{2}+\frac{a}{2}-\frac{1}{4}z^2-\left(\frac{a^2}{4}+\frac{a}{2}\right)\frac{1}{z^2}  \Big]y(z)=0.
\end{equation}
By the inversion formula, the crossover density $\rho_{c,a}(\lambda)$ is given as  
\begin{align}\label{Inversion}
\begin{split}
\rho_{c,a}(\lambda)&=\frac{1}{\pi} \lim\limits_{\ve \rightarrow 0}\Im \Big[ G(\lambda-i\ve)\Big]
\\
&=\frac{1}{c \pi } \frac{1}{|y(\lambda)|^2} \Big( \Im[ y'(\lambda)] \Re [y(\lambda)]- \Im [y(\lambda)] \Re [y'(\lambda)] \Big) ,
\end{split}
\end{align}
for $\lambda\in\R$.
Observe here that by \eqref{ODE y}, the term $(\Im y' \, \Re y- \Im y \, \Re y')$ is constant along the real line, having a vanishing derivative. Therefore, by \eqref{Kummer ftn}, we finally conclude that  
\begin{equation} \label{Gausslog}
\rho_{c,a}(\lambda)= \frac{1}{Z(c,a)}\,  e^{-\frac{1}{2}\lambda^2+a\log|\lambda|  } \, \left| U\left(\frac{c}{2},\frac{1-a}{2},-\frac{\lambda^2}{2}\right) \right|^{-2},
\end{equation}
where $Z(c,a)$ is a normalisation constant. 
This is the solution for the crossover density for the potential \eqref{Gauss+a}. Unfortunately we have bee unable to determine the normalisation analytically for general parameter values $a$ and $c$, except for $c=0$ or $a=0$ as shown below. For that reason, in the plots presented at the end of Section~\ref{Main} the normalisation has been computed numerically. 

In the particular case $c=0$ we observe that   
\begin{equation}
U(0,\gamma,z)\equiv 1,
\end{equation} 
see \cite[p.327]{olver2010nist}. Therefore, we obtain the expected extremal case that
\begin{equation}
\rho_{0,a}(\lambda)=\frac{1}{2^{(1+a)/2}\, \Gamma( \frac{1+a}{2})} \,  e^{-\frac{1}{2}\lambda^2+a\log|\lambda| },
\end{equation}
with the density being proportional to $e^{-V}$.
On the other hand, if $a=0$, we can recover the density in \cite{allez2012invariant}, due to the identity (see e.g., \cite[p.328]{olver2010nist})
\begin{equation}
D_{-c}(iz) =2^{-c/2} \, e^{z^2/4}\, U\left( \frac{c}{2}, \frac{1}{2}, -\frac{z^2}{2} \right) . 
\end{equation}

In Section \ref{Main} in Figure~\ref{Gausslog1} (resp., Figure~\ref{Gausslog2}) we show an example for the crossover density \eqref{Gausslog} with $a=1/2$ (resp., $a=-1/2$). For comparison in Figure~\ref{Gausslog0} the known case $a=0$ from \cite{allez2012invariant}, interpolating between Gauss' and Wigner's semi-circular distribution, is also given.

\end{example}

\begin{acknowledgements}
The authors gratefully acknowledge discussions and helpful suggestions of Trinh Khanh Duy, Adrien Hardy, Nam-Gyu Kang, Myl\`{e}ne Ma\"{i}da, Seong-Mi Seo, Pierpaolo Vivo and Oleg Zaboronski, as well as detailed comments by Yacin Ameur and Gaultier Lambert on a 	preliminary version of the paper. 
We also wish to express our gratitude to Jeongho Kim for several valuable comments concerning the numerical verifications.

%
\end{acknowledgements}


\bibliographystyle{abbrv}
\bibliography{RMTbib}

\end{document}